\documentclass{article}[11pt]

\usepackage{amsmath}
\usepackage{amssymb}
\usepackage{amsthm}
\usepackage{ifpdf}
\usepackage{wrapfig}
\usepackage{fullpage}
\usepackage{cite}
\usepackage{subcaption}
\usepackage{float}
\usepackage{enumitem}
\usepackage{graphicx}
\usepackage{fancyhdr}
\usepackage[noend]{algpseudocode}
\usepackage{algorithm}
\usepackage{algorithmicx}

\ifpdf

  \usepackage[pdftex]{epsfig}
  \usepackage[pdftex]{hyperref}

\else

    \usepackage[dvips]{epsfig}
    \newcommand{\href}[2]{#2}

\fi

\theoremstyle{definition}
\newtheorem{theorem}{Theorem}[section]
\newtheorem{lemma}[theorem]{Lemma}
\newtheorem{corollary}[theorem]{Corollary}

\newcommand{\para}[1]{\vspace*{.1cm}\noindent\textbf{#1.}}

\def\compactify{\itemsep=0pt \topsep=0pt \partopsep=0pt \parsep=0pt}
\let\latexusecounter=\usecounter
\newenvironment{itemize*}
  {\def\usecounter{\compactify\latexusecounter}
   \begin{itemize}}
  {\end{itemize}\let\usecounter=\latexusecounter}
\newenvironment{enumerate*}
  {\def\usecounter{\compactify\latexusecounter}
   \begin{enumerate}}
  {\end{enumerate}\let\usecounter=\latexusecounter}
\newenvironment{description*}
  {\begin{description}\compactify}
  {\end{description}}

\begin{document}

\title{\Large Full Tilt: Universal Constructors for General Shapes with Uniform External Forces\thanks{
A preliminary version of this paper appeared in the Proceedings of the 2019 ACM-SIAM Symposium on Discrete Algorithms.
This research was supported in part by National Science Foundation Grant CCF-1817602.}}

\author{Jose Balanza-Martinez
\and David Caballero
\and Angel A. Cantu
\and Luis Angel Garcia
\and Timothy Gomez
\and Austin Luchsinger
\and Rene Reyes
\and Robert Schweller
\and Tim Wylie 
}

\date{}
\clearpage\maketitle
\thispagestyle{empty}

\vspace*{-.5cm}
\begin{center}
Department of Computer Science\\University of Texas - Rio Grande Valley \\Edinburg, TX 78539-2999, USA\\
\{jose.balanzamartinez01, david.caballero01, angel.cantu01, luis.a.garcia01, timothy.gomez01, austin.luchsinger01, rene.reyes01, robert.schweller, timothy.wylie\}@utrgv.edu
\end{center}

\begin{abstract} \small\baselineskip=9pt We investigate the problem of assembling general shapes and patterns in a model in which particles move based on uniform external forces until they encounter an obstacle.  In this model, corresponding particles may bond when adjacent with one another.  Succinctly, this model considers a 2D grid of ``open" and ``blocked" spaces, along with a set of slidable polyominoes placed at open locations on the board.  The board may be tilted in any of the 4 cardinal directions, causing all slidable polyominoes to move maximally in the specified direction until blocked.  By successively applying a sequence of such tilts, along with allowing different polyominoes to stick when adjacent, tilt sequences provide a method to reconfigure an initial board configuration so as to assemble a collection of previous separate polyominoes into a larger shape.

While previous work within this model of assembly has focused on designing a specific board configuration for the assembly of a specific given shape, we propose the problem of designing \emph{universal configurations} that are capable of constructing a large class of shapes and patterns.
For these constructions, we present the notions of \emph{weak} and \emph{strong} universality which indicate the presence of ``excess'' polyominoes after the shape is constructed.  In particular, for given integers $h,w$, we show that there exists a strongly universal configuration with $\mathcal{O}(hw)$ $1 \times 1$ slidable particles that can be reconfigured to build any $h \times w$ patterned rectangle.  We then expand this result to show that there exists a weakly universal configuration that can build any $h \times w$-bounded size connected shape.  Following these results, which require an admittedly relaxed assembly definition, we go on to show the existence of a strongly universal configuration (no excess particles) which can assemble any shape within a previously studied ``drop'' class, while using quadratically less space than previous results.

Finally, we include a study of the complexity of motion planning in this model. We consider the problems of deciding if a board location can be occupied by any particle (occupancy problem), deciding if a specific particle may be relocated to another position (relocation problem), and deciding if a given configuration of particles may be transformed into a second given configuration (reconfiguration problem).   We show all of these problems to be PSPACE-complete with the allowance of a single $2\times 2$ polyomino in addition to $1\times 1$ tiles.  We further show that relocation and occupancy remain PSPACE-complete even when the board geometry is a simple rectangle if domino polyominos are included.  

\end{abstract}
\newpage
\setcounter{page}1

\section{Introduction} \label{Introduction}
The ``tilt'' model is an elegant and simple model of robotic motion planning and assembly proposed by Becker et. al.~\cite{BecDemFek2014RMPS} with foundations in classical motion planning.  The model consists of a 2D grid board with ``open'' and ``blocked'' spaces, as well as a set of slidable polyominoes placed at open locations on the board. At the micro or nano scale, individually instructing specific particles may be impossible. Thus, this model uses a global external force to give all movable particles the same instruction. This may be done through any external force such as a magnetic field or gravity. In this work we assume gravity is the global force, and this is where the term \emph{tilt} comes from. In the simplest form of the problem, the board may be tilted (as an external force) in any of the four cardinal directions, causing all slidable polyominoes to slide maximally in the respective direction until reaching an obstacle. These mechanics have proved to be interesting bases for puzzle games, as evidenced by the maximal movement of \cite{atomixGame} and the global movement signals in \cite{megaMazeGame}. By adding bonding glues on polyomino edges, as is done in self-assembly theory~\cite{Doty-2012a,nacoTileSurvey,Woods20140214}, the polyominoes may stick together after each tilt, enabling this model to be a framework for studying the assembly of general shapes.

In this work we propose a new type of problem:  the design of board configurations that are \emph{universal} for a class of general shapes or patterns.  That is, we are interested in designing a single board configuration that is capable of being reconfigured to assemble any shape or pattern within a given set of shapes or patterns by applying the proper sequence of tilts.  This problem is distinct from problems considered in prior work in which a given shape is assembled by encoding the particular shape into a specific board configuration (i.e. encoding the shape by way of placement of ``blocked'' locations and initial polyomino placement).  As an analogy, prior work has focussed on building a specific purpose machine for building copies of a target shape via a simple repeating tilt sequence.  Here, we propose to build something more akin to a computer printer in which any shape or pattern, provided sufficient fuel/ink/polyominoes, may be requested from the machinery, and the particular shape requested is encoded by a provided tilt sequence.

As our primary focus in this paper relates to the problem of reconfiguring board configurations, a natural computational question arises:  what is the complexity of deciding if a given board configuration may be reconfigured into a second target board configuration?
In \cite{BecDemFek2014PCRA}, the authors show that minimizing the number of tilts to reconfigure between two configurations is PSPACE-complete.  In this paper, we add to this growing understanding of reconfiguration complexity by showing that deciding if reconfiguration is possible is PSPACE-complete with $1\times 1$ movable tiles and a single $2\times 2$ movable polyomino.
A related problem, which asks if a given location can be occupied by any polyomino, has also been considered. This problem was shown to be NP-hard \cite{BecDemFek2017PCCAL} in the restricted case of $1\times 1$ polyominoes that never stick together (but is not known to be in NP).
We also extend this line of work by showing this problem is PSPACE-Complete if polyominoes larger than than $1 \times 1$ tiles are allowed. We also show that this problem remains PSPACE-comeplete with different constraints such as limiting the number of larger polyominos and the board complexity.

\subsection{Motivations}
The beauty of the Tilt model is its simplicity combined with its depth.  These features allow this model to be a framework for computation and assembly within a number of potential applications at various scales.  A few examples at the macro, micro, and nano-scale are as follows.

\paragraph{Macro-scale}
The simplicity of the Tilt model allows for implementation with surprisingly simple components.  For instance, Becker et al. have constructed a modular, reconfigurable board with both geometry and sliding components, which they have demonstrated with video walkthroughs~\cite{tiltVideoBecker}.  Further, the components allow for attaching magnets to implement bonding between components. These bonded components can be viewed as polyominoes, which can even be sorted within a tilt system~\cite{2DDiscPoly}. Even a basic set of Legos is capable of quickly implementing many of our proposed constructions.  These systems can be represented realistically with games like Tilt by Thinkgames~\cite{ThinkfunTilt} and the marble based Labyrinth~\cite{labyrinthGame}.

\paragraph{Micro-scale}
It has been shown how to control magneto-tactic bacterium via global signals to move bacteria around complex vascular networks via magnets~\cite{Becker2013FCMM}, and how the same type of bacteria can be moved magnetically through a maze~\cite{khali2016controlmp}. This has a plethora of medical uses, such as minimal invasive surgery, targeted drug payloads, subdermal micro-constructions, ``targeted delivery of hemotherapeutic agents or therapeutic genes into malignant cells while sparing healthy cells''~\cite{Sinha1909}, and ``medical intervention in targeted zones inaccessible to catheterization''~\cite{martel2009mri}.

\paragraph{Nano-scale}
Promising potential applications for Tilt Assembly systems may even be found in the emerging field of DNA nanotechnology.  DNA walkers have been engineered to traverse programmed paths along substrates such as 2D sheets of DNA origami~\cite{Zhou2015APN,doi:10.1021/acsnano.7b02693}.  Such walkers may be augmented with activation signals to drive the walker forward one step by way of a DNA strand-displacement reaction~\cite{strand-displacement:2011}.  By flooding the system with a given signal, the walker could be pushed to continuously walk forward until stopped by some form of blocked location.  By further implementing four such signal reactions (one for each cardinal direction), and adding a specifically chosen signal type at each stage or step of the algorithm, a set of these DNA walkers become a nanoscale implementation of the Tilt Assembly model.  If feasible, such an implementation implies our Tilt algorithms offer a novel technique for the construction of nanoscale shapes and patterns.

\begin{table*}[t]\small
\centering
    \begin{tabular}{| c | c | c | c | c |@{}c@{} | @{}c@{} | c |}
        \hline
        \textbf{Result} & \textbf{Shape} & \textbf{Universal} & \textbf{Tilts} & \textbf{Size} & \textbf{Bonding} & \textbf{Geometry} & \textbf{Theorem} \\
        & \textbf{Class} &&&& \textbf{Complexity} & \textbf{Complexity} & \\
        \hline
        Fixed Shape & DROP & No & $O(hw)$\protect\footnotemark & $O(h^3w^3)$ & 2 labels & Connected & Thm. 6 in \cite{BecFek2017TAMF} \\
         \hline
         \textbf{Universal Patterns}  & All & Strongly & $O(hwk)$ & $O(hwk)$ & $k$ labels & Simple &  Thm. \ref{thm:pat}\\
        \hline
        \textbf{Universal  Shapes}  & All & Weakly & $O(hw)$ & $O(hw)$ & 1 label & Simple & Thm. \ref{thm:gen_shapes}\\
        \hline
        \textbf{Universal Shapes}  & DROP & Strongly & $O(h^2w)$ & $O(h^2w)$ & 2 labels & Connected & Thm. \ref{thm:shuriken} \\
        \hline
    \end{tabular}
    \caption{An overview of the shape construction results. The \textbf{Result} is the type of constructor achieved, and the \textbf{Shape Class} refers to the types of shapes that can be built. The \textbf{Tilts} are the number of board tilts, or external forces, required to build the shape. The \textbf{Size} refers to the size of the board necessary with respect to the size of the $h \times w$ bounding box of the shape being built. The \textbf{Bonding Complexity} is the number of distinct particle types that can attach to each other and are necessary to build the shape. The $k$ labels needed for patterns is the number of desired types of particles in the pattern ($1\leq k \leq |S|$). The \textbf{Geometry Complexity} refers to the complexity of the open and blocked pieces of the board based on the hierarchy provided in Section~\ref{Prelims}.
    The \textbf{Theorem} refers to where this information is from.}
    \label{table:shapes}
\end{table*}

\begin{table*}[t]
\centering
    \begin{tabular}{| c | c | c | c | c |}
        \hline
        \textbf{Problem} & \textbf{Shapes} & \textbf{Geometry} & \textbf{Complexity} & \textbf{Theorem} \\
        \hline
	      Relocation/Occupancy & $1\times 1$ & Connected & NP-hard &  1 in \cite{BecDemFek2017PCCAL}\\
        \hline
        Reconfiguration Optimization & $1\times 1$ & Connected & PSPACE-complete &  10 in \cite{BecDemFek2017PCCAL} \\
        \hline
        \textbf{Relocation/Occupancy} & $1\times 1$, $2\times 2$ & Connected & PSPACE-complete &  \ref{thm:relocation}, \ref{corr:occupancy} \\
        \hline
        \textbf{Reconfiguration} & $1\times 1$, $2\times 2$ & Connected & PSPACE-complete &  \ref{thm:reconfiguration} \\
        \hline
        \textbf{Relocation/Occupancy} & $1\times 1$, $1\times 2$, $2\times 1$ & Rectangular & PSPACE-complete &  \ref{thm:PSPACEdomi}, \ref{thm:occupancy} \\
        \hline
    \end{tabular}
	\caption{An overview of the computational complexity results related to tilt assembly and our results. The \textbf{Problem} gives the computational question. The \textbf{Shapes} refer to the size of the polyominos that are allowed to move within the world. The \textbf{Geometry} column refers the complexity of the nonmovable tiles needed in the reduction. The \textbf{Complexity} refers to the computational complexity class that the problem was proven to be a member, and the \textbf{Theorem} is the reference to the result.}
    \label{table:complexity}
\end{table*}

\subsection{Our Contributions in Detail}

We first show the existence of a universal configuration for building any $h\times w$ bounded shape or pattern. The result utilizes simple geometry (the set of open spaces on the board has genus-$0$), only a single type of bonding particle, and is quadratically smaller in size than the corresponding non-universal construction from previous work~\cite{BecFek2017TAMF}.  Moreover, this is the first result in the literature that is capable of building any connected shape.  However, in the case of general shape construction, we say this system is only \emph{weakly} universal in that it allows for the inclusion of ``helper'' polyominoes that are not counted as part of the final shape as they do not stick to any other tile.  Our next result is for \emph{strong} universality in which only a single final polyomino of the desired shape is permitted.  In this case we achieve a restricted class of shapes termed ``Drop'' shapes which are shapes buildable by dropping $1\times 1$ polyominoes onto the outside of the shape from any of the four cardinal directions.  Previous non-universal work has focused on both constructing and identifying members of the Drop shapes class~\cite{BecFek2017TAMF}.
A summary of universal shape construction results, along with closely related prior work, is provided in Table~\ref{table:shapes}.

Our next set of results focuses on various decision problems within the tilt model.
The \emph{occupancy} problem asks if a given location on a board can be occupied by a polyomino.
The \emph{relocation} problem asks if a specific tile is able to be relocated to a given location.
The \emph{reconfiguration} problem asks whether a given board configuration may be reconfigured into a second given configuration.
We show these problems to be PSPACE-complete, even when polyominoes are restricted to be $1\times 1$ and a single $2\times 2$ piece that never stick to each other. We also show that the \emph{occupancy} and \emph{relocation} problems remain PSPACE-complete if we limit our board to a rectangular frame but allow multiple $1 \times 2$ and $2 \times 1$ dominos. Our proofs rely on a reduction from a 2-tunnel gadget network game with different gadget types that was recently proven to be PSPACE-complete by Demaine, Grosof, Lynch, and Rudoy~\cite{DBLP:conf/fun/DemaineGLR18}.  Previous work on occupancy has shown the problem to be NP-hard~\cite{BecDemFek2017PCCAL} even when restricted to $1 \times 1$ pieces that do not stick.  A closely related problem of computing the minimum number of tilts needed to reconfigure between two board configurations has been shown to PSPACE-complete for $1 \times 1$ non-sticking pieces~\cite{BecDemFek2017PCCAL}.  A summary of our complexity results, along with closely related prior work, is provided in Table~\ref{table:complexity}.

\addtocounter{footnote}{0}
\footnotetext{This technique permits ``pipelined'' construction for the creation of $n$ copies of the target shape in amortized $O(1)$ number of tilts per copy.}

\section{Preliminaries} \label{Prelims}

\para{Board}
A \emph{board} (or \emph{workspace}) is a rectangular region of the 2D square lattice in which specific locations are marked as \emph{blocked}.  Formally, an $m\times n$ board is a partition $B=(O,W)$ of $\{(x,y) | x\in \{1, 2, \dots, m\}, y\in \{1, 2, \dots, n\}\}$ where $O$ denotes a set of \emph{open} locations, and $W$ denotes a set of \emph{blocked} locations- referred to as ``concrete'' or ``walls.''  We classify the different possible board geometries according to the following hierarchy:

\begin{itemize}
    \item Connected:\footnote{In \cite{FullTilt2019}, this hierarchy level was known as \emph{complex}. We modify this class to allow for a proper hierarchy of shape classes.}  A board is said to have \emph{connected} geometry if the set of open spaces $O$ for a board is a connected shape.
    \item Simple:\footnote{\cite{FullTilt2019} defines simple geometry based on the connectivity of $W$. We also modify the concept for hierarchical purposes. }  A connected board is said to be \emph{simple} if $O$ has genus-0.
    \item Monotone:  A simple board is \emph{monotone} if $O$ is either horizontally monotone or vertically monotone.
    \item Convex:  A monotone board is \emph{convex} if $O$ is both horizontally monotone and vertically monotone.
    \item Rectangular:  A convex board is \emph{rectangular} if $O$ is a rectangle.
\end{itemize}

Our board definitions have changed since \cite{FullTilt2019} in order to create the hierarchy shown above.

\para{Tiles}
A tile is a labeled unit square centered on a non-blocked point on a given board. Formally, a tile is an ordered pair $(c,a)$ where \emph{c} is a coordinate on the board, and \emph{a} is an attachment label. Attachment labels specify which types of tiles will stick together when adjacent, and which have no affinity. For a given alphabet of labels $\Sigma$, and some \emph{affinity} function $G: \Sigma \times \Sigma \rightarrow \{0,1\}$ which specifies which pairs of labels attract ($G(a,b)=1$) and which do not ($G(a,b)=0$),  we say two adjacent tiles with labels $a$ and $b$ are \emph{bonded} if $G(a,b)=G(b,a) = 1$.
\\

\para{Polyomino}
A \emph{polyomino} is a finite set of tiles $P = \{t_1, \ldots t_k\}$ that is 1) connected with respect to the coordinates of the tiles in the polyomino and 2) \emph{bonded} in that the graph of tiles in $P$ with edges connecting bonded tiles is connected. A polyomino that consists of a single tile is informally referred to as simply a ``tile".
\\

\para{Configurations}
A configuration is an arrangement of polyominoes on a board such that there are no overlaps among polyominoes, or with blocked board spaces.  Formally, a configuration $C=(B, P=\{P_1\ldots P_k\})$ consists of a board $B$, along with a set of non-overlapping polyominoes $P$ that each do not overlap with the blocked locations of board $B$.
\\

\para{Step} A \emph{step} is a way to turn one configuration into another by way of a global signal that moves all polyominoes in a configuration one unit in a direction $d \in \{N,E,S,W\}$ when possible without causing an overlap with a blocked location, or another polyomino.  Formally, for a configuration $C=(B,P)$, consider the translation of all polyominoes in $P$ by 1 unit in direction $d$.  If no overlap with blocked board spaces occurs, then the new configuration is derived by first performing this translation, and then merging each pair of polyominoes that each contain one tile from a now (adjacently) bonded pair of tiles.  If an overlap does occur, for each polyomino for which the translation causes an overlap with a blocked space, temporarily add these polyominoes to the set of blocked spaces and repeat.  Once the translation induces no overlap with blocked spaces, execute the translation and merge polyominoes based on newly bonded tiles to arrive at the new configuration.  If all polyominoes are eventually marked as blocked spaces, then the step transition does not change the initial configuration.  If a configuration does not change under a step transition for direction $d$, we say the configuration is \emph{$d$-terminal}.  In the special case that a step causes a polyomino to ``leave the board", we simply remove the polyomino from the configuration.
\\

\para{Tilt} A \emph{tilt} in direction $d \in \{N,E,S,W\}$ for a configuration is executed by repeatedly applying a step in direction $d \in \{N,E,S,W\}$ until a $d$-terminal configuration is reached.  We say that a configuration $C$ can be \emph{reconfigured in one move} into configuration $C'$ (denoted $C \rightarrow_1 C'$) if applying one tilt in some direction $d$ to $C$ results in $C'$.  We define the relation $\rightarrow_*$ to be the transitive closure of $\rightarrow_1$. Therefore, $C \rightarrow_* C'$
means that $C$ can be reconfigured into $C'$ through a sequence of tilts.
\\

\para{Tilt Sequence} A \emph{tilt sequence} is a series of tilts which can be inferred from a series of directions $ D = \langle d_1, d_2,\dots, d_k \rangle$; each $d_i \in D$ implies a tilt in that direction. For simplicity, when discussing a tilt sequence, we just refer to the series of directions from which that sequence was derived. Given a starting configuration, a tilt sequence corresponds to a sequence of configurations based on the tilt transformation.  An example tilt sequence $\langle S, W, N, W, S, W, S \rangle$ and the corresponding sequence of configurations can be seen in Figure \ref{fig:simple_example}.
\\
\begin{figure}[H]
    \centering
	\begin{subfigure}[b]{0.1\textwidth}
        \includegraphics[width=1.\textwidth]{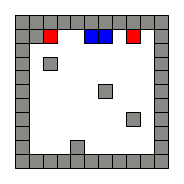}
        \caption{Start}
    \end{subfigure}
	\begin{subfigure}[b]{0.1\textwidth}
        \includegraphics[width=1.\textwidth]{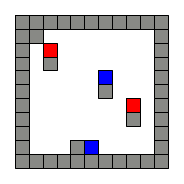}
        \caption{South}
    \end{subfigure}
    \begin{subfigure}[b]{0.1\textwidth}
        \includegraphics[width=1.\textwidth]{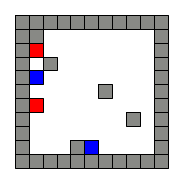}
        \caption{West}
    \end{subfigure}
    \begin{subfigure}[b]{0.1\textwidth}
        \includegraphics[width=1.\textwidth]{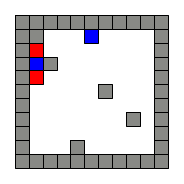}
        \caption{North}
    \end{subfigure}
        \begin{subfigure}[b]{0.1\textwidth}
        \includegraphics[width=1.\textwidth]{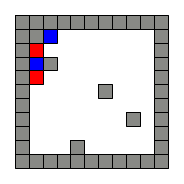}
        \caption{West}
    \end{subfigure}
        \begin{subfigure}[b]{0.1\textwidth}
        \includegraphics[width=1.\textwidth]{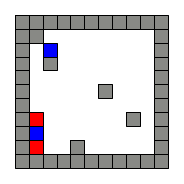}
        \caption{South}
    \end{subfigure}
        \begin{subfigure}[b]{0.1\textwidth}
        \includegraphics[width=1.\textwidth]{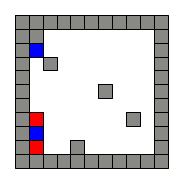}
        \caption{West}
    \end{subfigure}
        \begin{subfigure}[b]{0.1\textwidth}
        \includegraphics[width=1.\textwidth]{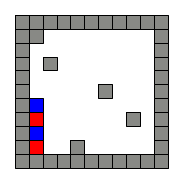}
        \caption{South}
    \end{subfigure}

    \caption{Tilt Example}
    \label{fig:simple_example}
\end{figure}

\para{Universal Configuration} A configuration $C'$ is universal to a set of configurations $\mathcal{C} = \{ C_1, C_2, \dots, C_k \}$ if and only if $C' \rightarrow_* C_i \ \forall \mkern9mu C_i \in \mathcal{C}$.
\\

\para{Configuration Representation}  A configuration may be interpreted as having constructed a ``shape" in a natural way.  Define a shape to be a connected subset $S\subset Z^2$.  A configuration \emph{strongly} represents $S$ if the configuration consists of a single polyomino whose tile coordinates are exactly the points of some translation of $S$.  A weaker version (discussed in detail in section \ref{sec:patterns}) allows for some ``helper" polyominoes to exist in the configuration and not count towards the represented shape. In this case, we say a configuration \emph{weakly} represents $S$.  We extend this idea of shape representation to include patterns. We say a configuration represents a pattern if each attachment label used in the representation corresponds to exactly one symbol of the pattern. For example, a configuration with attachment labels $\{a_1, a_2, \dots, a_n\}$ represents a pattern with symbols $\{s_1,s_2, \dots, s_n \}$ if the location of each tile $t_i \in C$ with attachment label $a_i$ matches with the positions of symbol $s_i$ in the pattern.
\\

\para{Universal Shape Builder}  Given this representation, we say a configuration $C'$ is \emph{universal} for a set of shapes $U$ if and only if there exists a set of configurations $\mathcal{C}$ such that 1) each $u \in U$ is represented by some $C \in \mathcal{C}$ and 2) $C'$ is universal for $\mathcal{C}$. If each $u \in U$ is strongly represented by some $C \in \mathcal{C}$, we say $C'$ is \emph{strongly universal} for $U$.  Alternately, if each $u \in U$ is weakly represented by some $C \in \mathcal{C}$, we say $C'$ is \emph{weakly universal} for $U$. In a similar way, a configuration can be universal for a set of patterns.

\section{Patterns and General Shapes} \label{sec:patterns}
In this section we present a shape builder which is universal for the set of all $h \times w$ binary-patterned rectangles.  We then extend this approach to achieve a universal constructor for any connected shape, or even any patterned connected shape.  We first cover a high-level overview of how the construction works and then formally state and prove the result.

\begin{figure}[ht!]
  \centering
	\begin{subfigure}[b]{.45\textwidth}
        \includegraphics[width=1.0\textwidth]{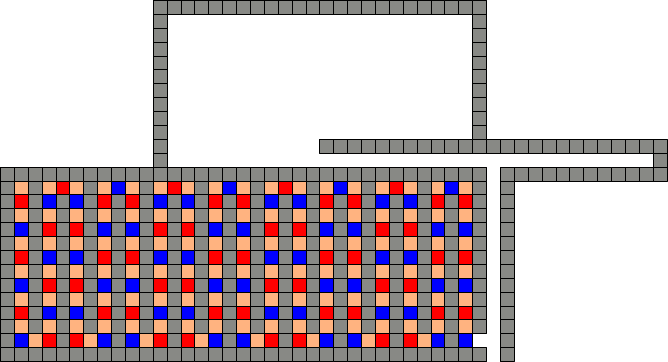}
        \caption{Starting Configuration}
        \label{fig:pat_start}
  	\end{subfigure}
    \hspace*{.1cm}
  	\begin{subfigure}[b]{.45\textwidth}
        \includegraphics[width=1.0\textwidth]{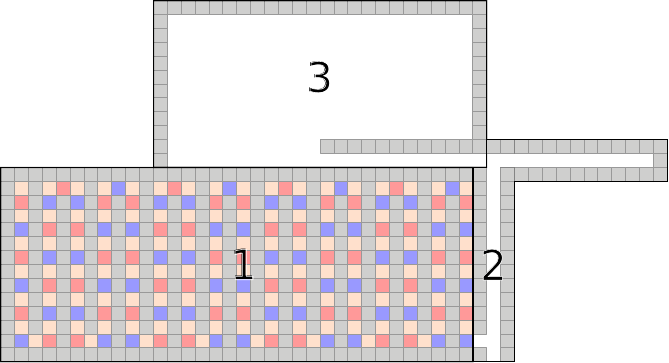}
        \caption{Sections}
        \label{fig:pat_sections}
  	\end{subfigure}
    \caption{An overview of the universal pattern and shape constructor. (a) The universal binary-patterned rectangle builder configuration in a starting configuration. (b) We refer to various sections of the configuration by different names. Section 1 is the \emph{fuel chamber}, section 2 is the \emph{loading chamber}, and section 3 is the \emph{construction chamber}.}
    \label{fig:pat_overview}
\end{figure}

\subsection{Binary-patterned Rectangles: Construction}
The high-level idea behind this construction is simple: tiles are removed from the fuel chamber one at a time and are either used in ``line assembly,'' or ejected from the system. Once a patterned line is complete, it is used for ``rectangle assembly''. Essentially, we are assembling a patterned rectangle pixle-by-pixel, row-by-row.
For a given rectangle size ($h \times w$), we can construct a tilt assembly configuration (Figure \ref{fig:pat_overview}) which can assemble any binary-patterned rectangle of that size. 

\begin{figure}[ht!]
  \centering
  \begin{subfigure}[b]{.22\textwidth}
      \includegraphics[width=1.0\textwidth]{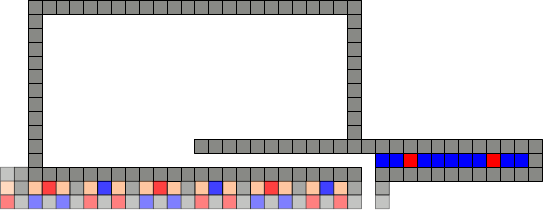}
      \caption{}
      \label{fig:pat_hl1}
  \end{subfigure}
    \hspace*{.1cm}
  \begin{subfigure}[b]{.22\textwidth}
      \includegraphics[width=1.0\textwidth]{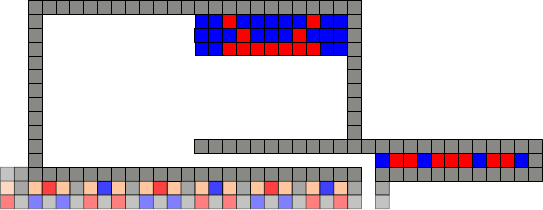}
      \caption{}
      \label{fig:pat_hl2}
  \end{subfigure}
    \hspace*{.1cm}
  \begin{subfigure}[b]{.22\textwidth}
      \includegraphics[width=1.0\textwidth]{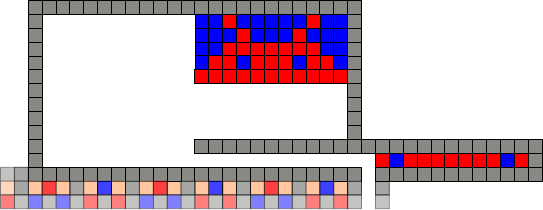}
      \caption{}
      \label{fig:pat_hl3}
  \end{subfigure}
  	\hspace*{.1cm}
	\begin{subfigure}[b]{.22\textwidth}
	    \includegraphics[width=1.0\textwidth]{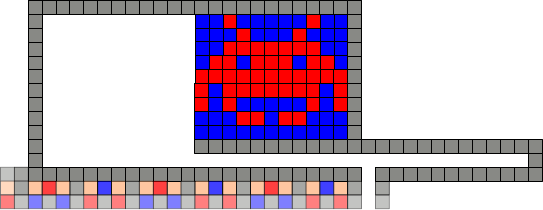}
	    \caption{}
		\label{fig:pat_hl4}
	\end{subfigure}
    \caption{The rectangle construction process at different points. Patterns are built row by row and then added to the final shape. (a) depicts an assembled row of the pattern. (b) and (c) show subsequent configurations after more rows have been added and (d) shows the final assembled pattern.}
    \label{fig:pat_shapes}
\end{figure}

\paragraph{Fuel Chamber}  The fuel chamber is the portion of our system which contains the tiles that are used for construction. Our construction utilizes 2 types of tiles that attach to themselves and to each other, and 1 tile that does not attach to anything (we often refer to this tile as \emph{sand})\footnote{The term \emph{sand} is inspired by the game Minecraft~\cite{minecraft} in which blocks of type ``sand'' do not stick to adjacent blocks and will fall freely if nothing lies beneath them.  Here, sand refers to $1\times 1$ tiles that do not stick to any other tile and are not counted as part of the final assembly.}.  We must be able to fill up the entire $hw$-sized porion of the construction chamber with either of the ``sticky'' tiles, so we need to allocate room for that many tiles in the fuel chamber. In the fuel chamber, each sticky tile must be separated by a sand tile, so we do not end up with ``clumps'' of fuel.  Thus, we have $4hw$ tiles in our fuel chamber. To progress fuel through the fuel chamber, an input sequence of $\langle N,E,S \rangle$ is required. It is important to note that each of our commands in Table~\ref{tab:pat_sequences} ends with this sequence, thereby naturally advancing the fuel through the fuel chamber.

\paragraph{Loading Chamber}  The loading chamber is the next portion of the construction. It is here that tiles can either be removed from the configuration, or loaded into the upper part of the loading chamber. As tiles are added to the loading chamber, one row (line) of the pattern is built. The vertical portion of this chamber is the same height as the fuel chamber, and the horizontal portion has a width of $w$. Once a line is complete, it is added to the construction chamber.

\paragraph{Construction Chamber}  The construction chamber is the portion of the tilt assembly system where our pattern will be constructed. Its dimensions are determined by the size of the rectangle to be constructed.
The upper portion (with dimensions $h \times 2w$) is where the binary-patterned rectangle will be constructed. The lower portion (with dimensions $1 \times 2w$) is where the pre-assembled lines enter from the loading chamber.
The construction chamber can be thought of as having two parts. The left portion is where lines are added to the already existing pattern, and the right portion is where the pattern is stored for other move sequences.

\begin{figure}[t!]
  \centering
  	\begin{subfigure}[b]{.25\textwidth}
        \includegraphics[width=1.0\textwidth]{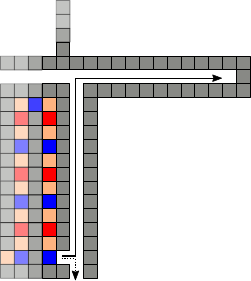}
        \caption{Add/Remove Tile}
        \label{fig:pat_add}
  	\end{subfigure}
    \hspace*{.15cm}
  	\begin{subfigure}[b]{.575\textwidth}
      \centering
  	    \includegraphics[width=1.0\textwidth]{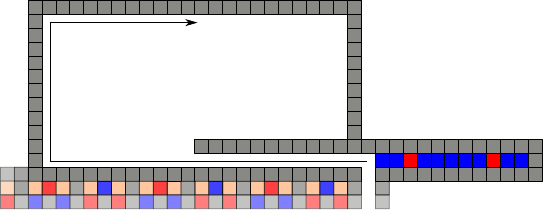}
  	    \caption{Add Line}
  		\label{fig:pat_remove}
  	\end{subfigure}
    \caption{Basic sequences used in the constructor. The cyclic sequence to advance the fuel is $\langle S,E,N,E,S,E, \dots \rangle$. (a) The sequence to add a tile to the line is $\langle E,N,E,S \rangle$. (b) The sequence to remove a tile from the system is $\langle E,S,N,E,S \rangle$.}
    \label{fig:pat_fuel}
\end{figure}

\paragraph{Line Construction}  The $1 \times w$ lines (rows of the rectangle) are constructed pixel-by-pixel (tile-by-tile) from right-to-left. For each pixel, the user decides if a blue or red tile should be placed and discards the other (as well as the separating sand). The sequences required for tile addition and removal are in Figure \ref{fig:pat_fuel}.

\paragraph{Rectangle Construction}  Once the line (row) is complete, the next step is to add it to the construction chamber. The line construction process is then repeated over and over again, with each new line being added to the construction chamber. Thus, the rectangle is built row-by-row from top-to-bottom. The sequences for all moves are in Table \ref{tab:pat_sequences}.

\begin{table}[ht]
  \centering
\begin{tabular}{|c|c|c|c|}
	\hline
	\textbf{Moves} & Add tile & Remove tile & Add line \\
	\hline
	\textbf{Tilts} &  $\langle E,N,E,S \rangle$ & $\langle E,S,N,E,S \rangle$ & $\langle W,N,E,S \rangle$ \\
  \hline
\end{tabular}
\caption{Tilt sequences used for general shape and pattern construction.
Note that a tile will not be removed during the Add Tile and Add Line commands (despite the presence of $\langle E, S \rangle$ in both) because of the preceding $N$ command and the fact that the rightmost column of the fuel chamber will be missing at least 1 tile.}
\label{tab:pat_sequences}
\end{table}

\subsection{Patterned Rectangles and General Shapes: Formal Results}

\begin{theorem}\label{thm:pat}
  Given two positive integers $h,w \in \mathbb{Z}^+$, there exists a configuration $C$ which is \emph{strongly universal} for the set of patterns  $U = \{ u \ | \ u \ \text{is an} \ h \times w \ \text{rectangle, where each pixel has a label} \ x \in \{ 0,1 \}  \}$. This configuration has size $\mathcal{O}(hw)$ and uses $\mathcal{O}(hw)$ tilts to reconfigure into a configuration which strongly represents any pattern $u \in U$.
\end{theorem}

\begin{proof}
	We show this by constructing a configuration $C = (B,P)$, similar to that shown in Figure~\ref{fig:pat_start}, where $B$ is the board and $P$ is the set of $1\times 1$ polyominoes in the fuel chamber in their starting configuration.

  \noindent The three chambers of the board (fuel, loading, and construction) with scale based on $h$ and $w$ are:

  \begin{itemize}
      \item \textit{Fuel chamber.}
      Let the fuel chamber of this configuration be of height $h + 3$ and length $8w + 2$. The fuel chamber is a rectangular area with staggered columns of blocked locations, so as to create a serpentine path of length $4hw$. Let this path be filled with tiles of alternating attachment labels $a,b$ buffered with label $\varepsilon$, where $G(a,a) = 1, G(a,b) = 1, G(b,b) = 1$ and $\varepsilon$ is the attachment label with no affinity. This allows enough room to have a single serpentine line of $4hw$ tiles that do not stick to each other.

      \item \textit{Loading chamber.}
      The loading chamber must have a vertical hallway whose height is the same as the height of the fuel chamber ($h + 2$). Let the loading chamber also contain a horizontal tunnel whose width is $w + 2$.

      \item \textit{Construction chamber.}
      The construction chamber can be described as a large set of open locations with a perimeter of blocked locations resembling the figures above. Let the construction chamber have height $h + 2$ and width $2w + 2$.

  \end{itemize}
  By observing the dimensions of the chambers, the height of the board is $2h + 5$ and the width of the board is $9w + 4$. So, the size of the board is $\mathcal{O}(hw)$.

  \textit{Reconfiguration.}
  Consider the set of configurations $\mathcal{C}'$, where each $c' \in \mathcal{C}'$ is similar to that shown in Figure \ref{fig:pat_start} and \emph{strongly} represents some $u \in U$ by a patterned-rectangle in the construction area. From configuration $c$, and the construction shown above, there exists a tilt sequence to construct any binary-colored line pixel-by-pixel. There also exists a tilt sequence to construct a rectangle with these binary-colored lines row-by-row. These sequences, along with that to discard any superfluous tiles, shows that there exists a complete sequence to build any rectangular binary pattern. Thus,  $\forall$ $c' \in \mathcal{C}', C \rightarrow_* c'$. 
\end{proof}

\paragraph{General Patterns}
By increasing the size of the fuel chamber, we can easily generalize this result to $k$ tile types (represented by labels or colors). Given $k$ different labeled tiles for the pattern, the fuel chamber just needs to repeat the sequence of tiles in order, with sand in between each labeled tile, enough times to ensure the desired pattern can be built.

\begin{corollary}\label{thm:pat_gen}
  Given three positive integers $h,w,k \in \mathbb{Z}^+$, there exists a configuration $C$ which is \emph{strongly universal} for the set of patterns $U = \{ u \ | \ u \ \text{is an} \ h \times w \ \text{rectangle, where each pixel has a label} \ x \in \{ 0,1, \dots, k \}\}$. This configuration has size $\mathcal{O}(hwk)$ and uses $\mathcal{O}(hwk)$ tilts to reconfigure into a configuation which strongly represents any pattern $u \in U$.
\end{corollary}

\begin{proof}
  We use the construction from Theorem \ref{thm:pat} and extend the fuel chamber to include $k$ colors. The length of the serpentine fuel line would then be $\mathcal{O}(hwk)$ meaning the width of our board would now scale with $k$ as well. Also, since each pixel requires the user to select a color and discard the others, the number of tilts would also scale with $k$. 
\end{proof}

\paragraph{General Shapes} With a simple extension of this result, we are able to achieve the construction of general shapes. By including sand in the building process and the finished pattern, we can construct any connected shape. This is weakly built by our definition since there are non-attached tiles built along with it. Figure \ref{fig:cat_shapes} shows this process.

\begin{theorem}\label{thm:gen_shapes}
  Given two positive integers $h,w \in \mathbb{Z}^+$, there exists a configuration $C$ which is \emph{weakly universal} for the set of shapes $U = \{ u | u \subseteq \{ 1,\dots,h\} \times \{ 1,\dots,w\} \}$. This configuration has size $\mathcal{O}(hw)$ and uses $\mathcal{O}(hw)$ tilts to reconfigure into a configuration which weakly represents any shape $u \in U$.
\end{theorem}

\begin{proof}
  Following Theorem \ref{thm:pat}, we know that we can construct any rectangular binary pattern.  By removing one of the labeled (``sticky'') tile types, we can create a binary pattern using only one labeled type and the $\varepsilon$ (sand) tiles.  Then the only connected portion of the shape is the parts with the one labeled tile type. Hence, we can build a connected shape surrounded by ``sand.'' This process results in a shape builder that is universal for the set of polyomino shapes that fit within an $h \times w$ bounding box. 
\end{proof}

\begin{figure}[H]
  \centering
  \begin{subfigure}[b]{.22\textwidth}
      \includegraphics[width=1.0\textwidth]{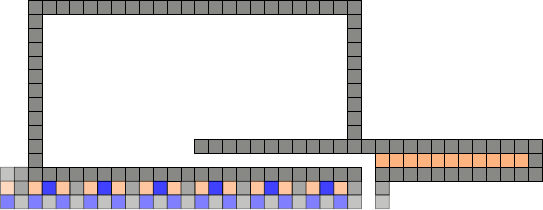}
      \caption{}
      \label{fig:hl1}
  \end{subfigure}
    \hspace*{.1cm}
  \begin{subfigure}[b]{.22\textwidth}
      \includegraphics[width=1.0\textwidth]{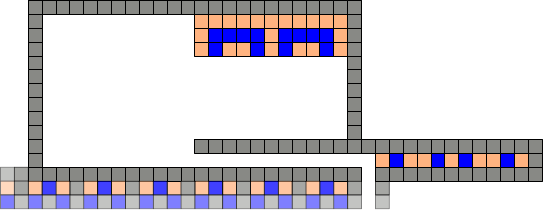}
      \caption{}
      \label{fig:hl2}
  \end{subfigure}
    \hspace*{.1cm}
  \begin{subfigure}[b]{.22\textwidth}
      \includegraphics[width=1.0\textwidth]{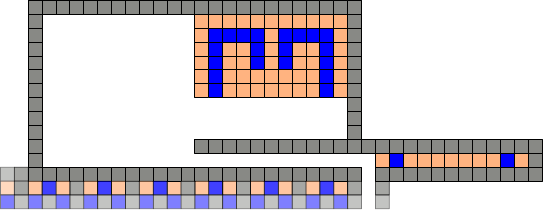}
      \caption{}
      \label{fig:hl3}
  \end{subfigure}
  	\hspace*{.1cm}
	\begin{subfigure}[b]{.22\textwidth}
	    \includegraphics[width=1.0\textwidth]{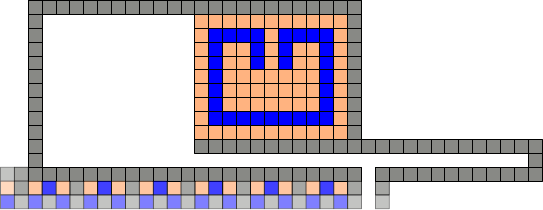}
	    \caption{}
		\label{fig:hl4}
	\end{subfigure}
    \caption{The shape construction process at different intervals. (a) The first row of the shape being built, however, the shape does not need this row and so only sand is added to the row. (b-c) Partially built shape with several disconnected sections of the shape being built simultaneously. (d) The finished shape with sand encasing the shape.}
    \label{fig:cat_shapes}
\end{figure}

\paragraph{Patterned Shapes} By keeping any number of labeled tiles, we can also build any patterned connected shape.

\begin{corollary}\label{thm:pat_shapes}
  Given three positive integers $h,w,k \in \mathbb{Z}^+$, there exists a configuration $C$ which is \emph{weakly universal} for the set of shapes $U = \{ u | u \subseteq \{ 1,\dots,h\} \times \{ 1,\dots,w\}, \ \text{and each pixel in u has a label} \ x \in \{ 0,1,\dots,k \}  \}$.  This configuration has size $\mathcal{O}(hwk)$ and uses $\mathcal{O}(hwk)$ tilts to reconfigure into a configuration which weakly represents any shape $u \in U$.
\end{corollary}

\begin{proof}
  By the same argument as Theorem \ref{thm:pat_gen}, we can add $k$ colors to our fuel chamber, causing the board size and number of tilts to scale with $k$. 
\end{proof}

\subsection{Additional Notes}
\paragraph{Keeping the Undesirable Tiles}
One aspect of the constructor that may be undesirable or infeasible for some pracitcal implemenations is the idea of removing pieces from the board. Although not shown explicitly, the constructor (and subsequently the tilt sequences) can easily be modified to handle the removed tiles in numerous alternate ways.
\begin{itemize}
    \item Trash. The easiest solution is to have a chamber that the tiles are thrown down into such that they become trapped (this can be achieved with a large enough chamber with one opening at the center of the top wall). The chamber must be large enough to accommodate pieces arbitrarily sticking together, and still prevent tiles from returning to the constructor.
    \item Only Sand as Trash. Each labeled (colored) sticky tile can have its own fuel chamber still using sand to separate each piece. Each fuel chamber would be a $2n^2 \times 1$ vertical chamber that requires a unique tilt sequence to get a piece of fuel of that type out. The only trash is then sand, so the trash chamber would also only need to be $h \times w$ large.
    \item Reusing tiles. With the standard fuel tank, the unwanted tile could be routed down and around to the back of the fuel chamber to be reinserted. This recycling would only require an extra chamber of sand to put between labeled pieces that were recycled consecutively because the sand in between them was used in the shape.
\end{itemize}

\paragraph{Freeing the shape}
Another possible drawback of the constructor as presented is that the shape is trapped in the constructor itself. The constructor may be easily modified to allow for one shape to be released and another constructed.
By enlarging the construction chamber to have a height of $2h$, and removing the top $h$ tiles from the left wall, we can create an opening for the finished shape to leave the constructor. Thus, the \emph{remove shape} tilt sequence $\langle N,W \rangle$ could be added, as it does not overlap with any of the pre-existing tilt sequences.

\section{Drop Shapes} \label{sec:Shapes}
Next, we consider a class of shapes discussed in \cite{BecFek2017TAMF} which we refer to as \emph{drop shapes}. A \emph{drop shape} is any polyomino which is constructable by adding particles from any of the four cardinal directions \emph{\{N, E, S, W\}} towards a fixed seed. For a formal definition of this class of shapes, see constructable polyominoes for the Tilt Assembly Problem (TAP) in \cite{BecFek2017TAMF}.
Figure \ref{fig:drop_shapes} shows an example of a valid and invalid drop shape.
Here, we present a shape builder which is strongly universal for the set of drop shapes.

\begin{figure}[H]
\centering
	\begin{subfigure}[b]{0.16\textwidth}
		\includegraphics[width=1.\textwidth]{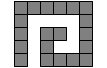}
		\caption{Valid}
	\end{subfigure}
	\begin{subfigure}[b]{0.16\textwidth}
		\includegraphics[width=1.\textwidth]{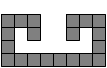}
		\caption{Invalid}
	\end{subfigure}
	\caption{Drop Shapes examples.  (a) This polyomino is buildable with the drop shape method, whereas (b) is a polyomino that is not a constructable drop shape. From a fixed single tile seed, it is not possible to build (b) by adding one tile at a time from a cardinal direction.}
	\label{fig:drop_shapes}
\end{figure}

\subsection{Universal Drop Shape Builder: Construction}
Figure \ref{fig:shuriken_overview} is an example of our universal drop-shape builder for polyominoes fitting within a $4\times 4$ bounding box. At a high-level, this construction works by following five phases, which are indicated by the labeled areas in Figure \ref{fig:shuriken_sections}.
\begin{enumerate}
    \item Select a red or blue tile from the fuel chamber. Area 1 shows the two types of fuel that stick to each other. Here, we use a sequence to pick the one we want.
    \item Choose which direction to add the tile from. We move the shape into the appropriate $N,S,E,$ or $W$ location and move the tile to the side we are adding from. As the area 2 labels show, we can move the new tile to the appropriate side of the shape.
    \item Choose which column/row to add the tile on the shape. This example is for any $4\times 4$ shape, so we choose columns (N/S) or row (E/W) to shoot the tile onto the shape where it will be added.
    \item This area is where the shape is held when the new tile is added. There is a different holding area for each of the four directions.
    \item In the last phase we move the shape to area 5 where it is held while we get the next tile to add. This holding chamber ensures the polyomino being built does not interfere with getting the next tile ready.
\end{enumerate}

A full overview of the process with a flowchart and the necessary tilt sequences is shown in Figure \ref{fig:Automaton} and Table \ref{tab:droptilts}, respectively.


\begin{figure*}[t]
  \centering
	\begin{subfigure}[b]{.45\textwidth}
        \includegraphics[width=1.0\textwidth]{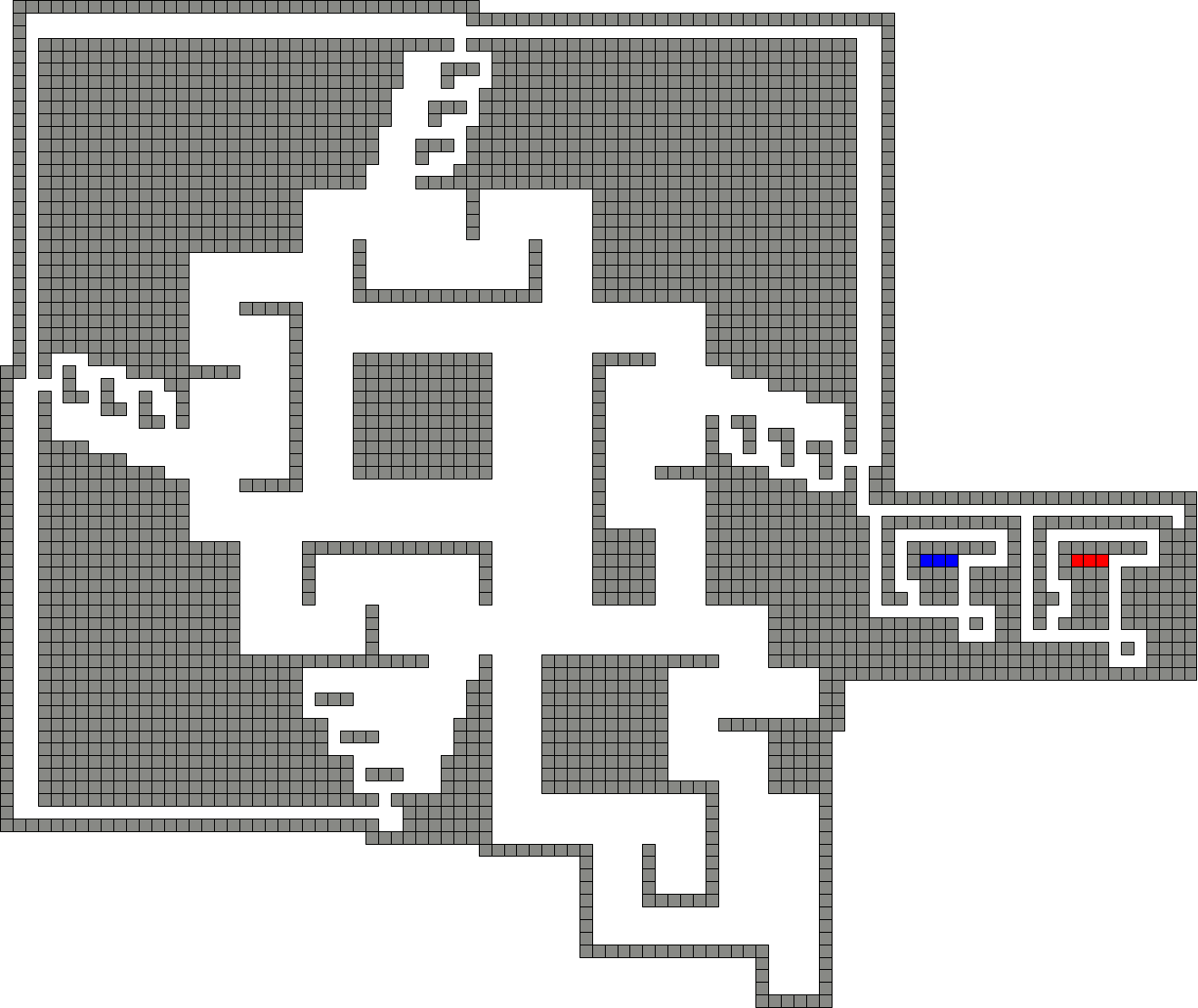}
        \caption{Starting Configuration}
        \label{fig:shuriken_start}
  	\end{subfigure}
		\hspace{.2cm}
		\begin{subfigure}[b]{.45\textwidth}
	        \includegraphics[width=1.0\textwidth]{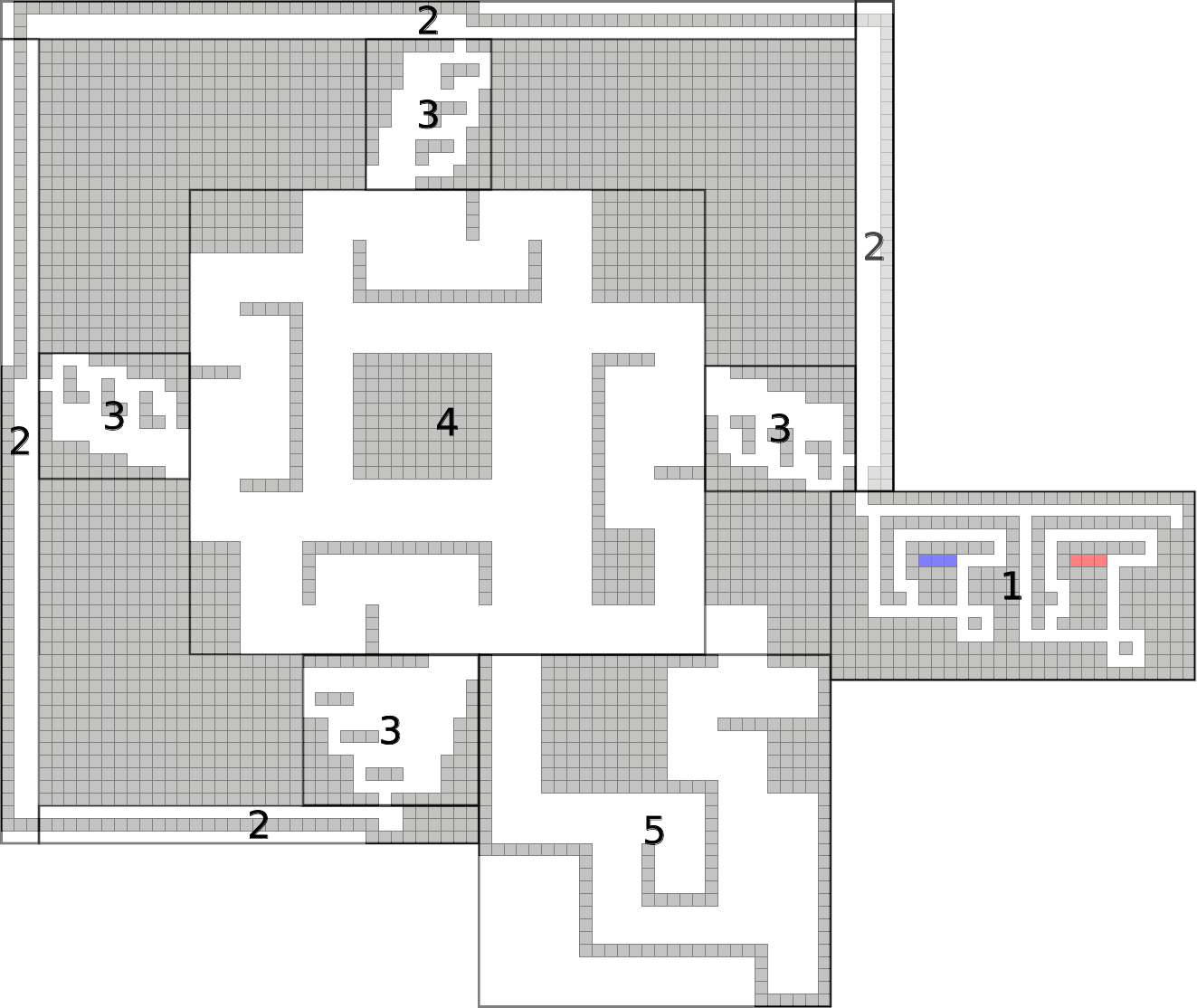}
	        \caption{Sections}
	        \label{fig:shuriken_sections}
	  	\end{subfigure}
    \caption{(a) The universal drop-shape builder. (b) An overview of the parts of the drop shape builder. Area 1 is the \emph{fuel chamber}, area 2 is the \emph{selection chamber}, the area 3's are the \emph{alignment chambers}, area 4 is the \emph{construction chamber}, and area 5 is the \emph{holding chamber}.}
    \label{fig:shuriken_overview}
\end{figure*}

\begin{figure}[H]
\centering
	\includegraphics[width=.45\textwidth]{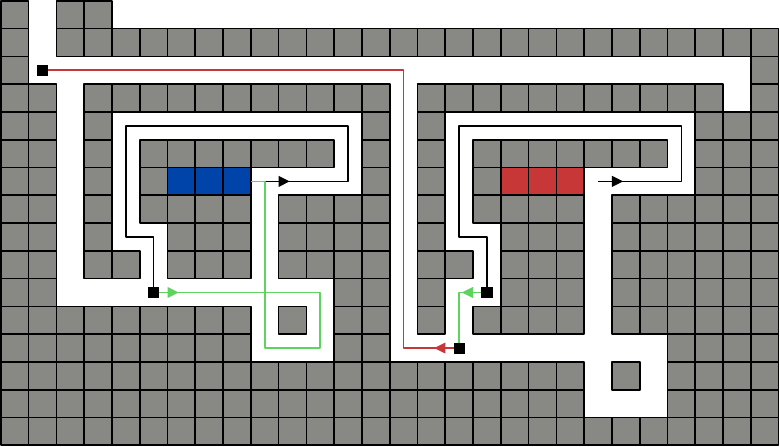}
	\caption{Tile selection gadget. Each tile is pulled out with the sequence $\langle E,N,W,S,E,S \rangle$, and stops at the first square. Then the left tile type (blue) is either pulled out of the gadget or put back in the storage area. This shows it being added back to the storage with $\langle E,S,W,N,W,S \rangle$. This sequence puts the next tile type (red) in a decision location. The red tile is selected with the sequence $\langle W,N,W,S,W,S \rangle$.
	}
	\label{fig:sec1}
\end{figure}
\paragraph{Fuel Chamber: Section 1}
The fuel chamber consists of one or more tile reservoirs. Each reservoir contains a tile type with the same attachment label. The tilt sequence $\langle E,N,W,S,E,S \rangle$ extracts one fuel tile out of each reservoir. After extraction, we select which tile we want one at a time. For each extracted tile, we can either perform a tilt sequence to return the fuel tile to its reservoir, or to select it to be used in shape construction. Figure \ref{fig:sec1} shows an example of selecting the second type of fuel. Once a fuel tile has been selected with the sequence $\langle W,N,W,S,W,S \rangle$, the storing tilt sequence must be performed for the rest of the fuel tiles that were not selected. To store a fuel tile the tilt sequence is $\langle E,S,W,N,W,S \rangle$. Note that the position of a selected fuel piece remains the same after an application of the storage sequence.

On all iterations after the first, we must also consider the position of the polyomino/assembly which is being constructed. Before executing the tilt sequence to extract a tile, the polyomino is located in the northmost eastmost notch of the holding chamber. A quick observation will note that the tilt sequence required to traverse the holding chamber is identical to the tile extraction sequence. Also note that the position of the assembly after an extraction, like the selected fuel piece, is invariant after each application of the storage sequence.  When the extracted fuel tile is to enter the selection chamber, the assembly will be ready to enter the construction chamber.

\paragraph{Selection Chamber: Section 2}
After selecting the fuel tile and storing the rest of the extracted tiles, we now move the fuel tile into the selection chambers. This stage of the construction simply selects the direction of attachment. After tilting $\langle N \rangle$, the fuel enters the eastern selection chamber and the assembly enters the construction chamber. To select an attachment from the east, perform the tilt sequence $\langle E,S,W \rangle$ (notice that this will position the assembly to the west of the eastern alignment chamber), else tilt $\langle W \rangle$ to select the next direction. After this western tilt, to select an attachment from the north, perform the tilt sequence $\langle N,E,S \rangle$ (notice this also positions the assembly below the northern alignment chamber), else tilt $\langle S \rangle$ to select the next direction. Now, after this southern tilt, to attach from the west perform the tilt sequence $\langle W,N,E \rangle$ (this too positions the assembly to the east of the western alignment chamber), else to select the last direction tilt $\langle E \rangle$. At this point, the fuel must be attached from the south. The sequence $\langle S,W,N \rangle$ will prepare the fuel tile for an attachment from the south (while also positioning the assembly to the north of the southern alignment chamber).

\paragraph{Alignment Chambers: Sections 3}
After an attachment direction has been chosen, the tile enters a gadget to select the row/column of the polyomino to target with the new tile. Figure \ref{fig:sec3} shows an example for a northern selection gadget for a $4\times 4$ polyomino. The tilt sequence to align with a particular location varies with respect to the direction of attachment. If the rightmost pixel of a northern attachment is chosen, the tilts $\langle W,S \rangle$ would be performed. Each successive column is chosen by performing the sequence $\langle E,S,W,S \rangle$. This sequence is repeated, each time moving tile further down the gadget, until the desired location is reached, at which point the attachment sequence $\langle W,S \rangle$ is executed. The alignment sequence for each attachment chamber does not affect the position of the assembly (it will remain positioned in the attachment chamber). Note the southern alignment gadget is slightly different but serves the same purpose. This difference is so the alignment and extraction sequences do not overlap.

\begin{figure}[H]
\centering
	\includegraphics[width=.18\textwidth]{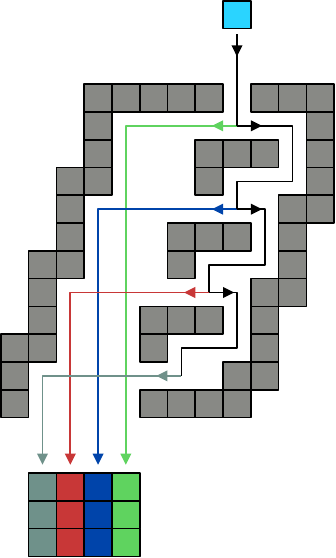}

	\caption{The column selection gadget for drop shapes. Assuming the shape to build is at a fixed location, this gadget allows any column to be selected to drop the new tile onto. The number of columns to drop from in this gadget determines the size of the shape we can build. Thus, this is for a drop shape within a $4\times 4$ bounding box. This gadget is repeated on each of the four sides of the drop-shape constructor (with a slightly modified one on the south side).
	}
	\label{fig:sec3}
\end{figure}

\paragraph{Construction Chamber: Section 4}
This central chamber is where the constructed assembly will be housed as it is being assembled. As we saw in the selection paragraph, the sequences required to position the assembly in front of a particular alignment chamber are the same as the sequences required to prepare the fuel piece to attach from that chamber. This, along with the alignment chamber process, allows us to pinpoint the attachment location of a fuel piece to the assembly.

\paragraph{Holding Chamber: Section 5}
Once the new tile has been added, we return the assembly to the holding chamber, which is where the assembly resides while the next fuel tile is being selected (Section 1). Once again, the target location in the holding chamber is the northmost westmost notch. As mentioned, the holding chamber's tilt sequences's are identical to that of the fuel chamber's. This ensures that we can select a fuel type without repositioning the assembly.

\begin{table*}
\centering
\begin{tabular}{|c|c|}\hline
    $s_{0}$ &  $\langle E,N,W,S,E,S \rangle + \langle E,S,W,N,W,S \rangle^{i} + \langle W,N,W,S,W,S \rangle + \langle E,S,W,N,W,S \rangle^{j}$ \\ \hline
    $s_{\text{E}1}$ & $\langle N,E,S,W,S\rangle$  \\ \hline
    $s_{\text{N}1}$ &  $\langle N,W,N,E,S\rangle$\\ \hline
    $s_{\text{W}1}$ & $\langle N,W,S,W,N,E\rangle$ \\ \hline
    $s_{\text{S}1}$ &   $\langle N,W,S,E,S,W,N\rangle$  \\ \hline
    $s_{\text{E}2}$ &   $\langle S,W,N,W\rangle ^{j} +  \langle N, W, S, E, S, E, S\rangle$ \\ \hline
    $s_{\text{N}2}$ &  $\langle E,S,W,S\rangle ^{j} +  \langle W, S, E, N, E, S, E, S\rangle$ \\ \hline
    $s_{\text{W}2}$ &  $\langle N,E,S,E\rangle ^{j} +  \langle S, E, N, W, E, S, E, S, E, S\rangle$ \\ \hline
    $s_{\text{S}2}$ & $\langle W,N \rangle ^{j} +  \langle E, N, W, S, W, N, E, S, E, S\rangle$ \\ \hline
\end{tabular}
    \caption{The sequence of tilts used in the flowchart (Figure \ref{fig:Automaton}). $s_0$ denotes the tile selection from $i+j+1$ different tile types and chambers. The sequence also places the current polyomino in the correct area for the next sequences (Figure \ref{fig:Launch}). For area 2 in Figure \ref{fig:shuriken_sections}, $s_{\text{E}1}$, $s_{\text{N}1}$, $s_{\text{W}1}$, $s_{\text{S}1}$ represents choosing to attach the new tile from the east, north, west, or south side, respectively. Similarly, the alignment selection (area 3), from the east, north, west, or south side, is represented by $s_{\text{E}2}$, $s_{\text{N}2}$, $s_{\text{W}2}$, and $s_{\text{S}2}$, respectively. Note there may be up to $j-1$ columns.}\label{tab:droptilts}
\end{table*}

\begin{figure*}
  \centering
	\begin{subfigure}[b]{.48\textwidth}
       		 \includegraphics[width=1.0\textwidth]{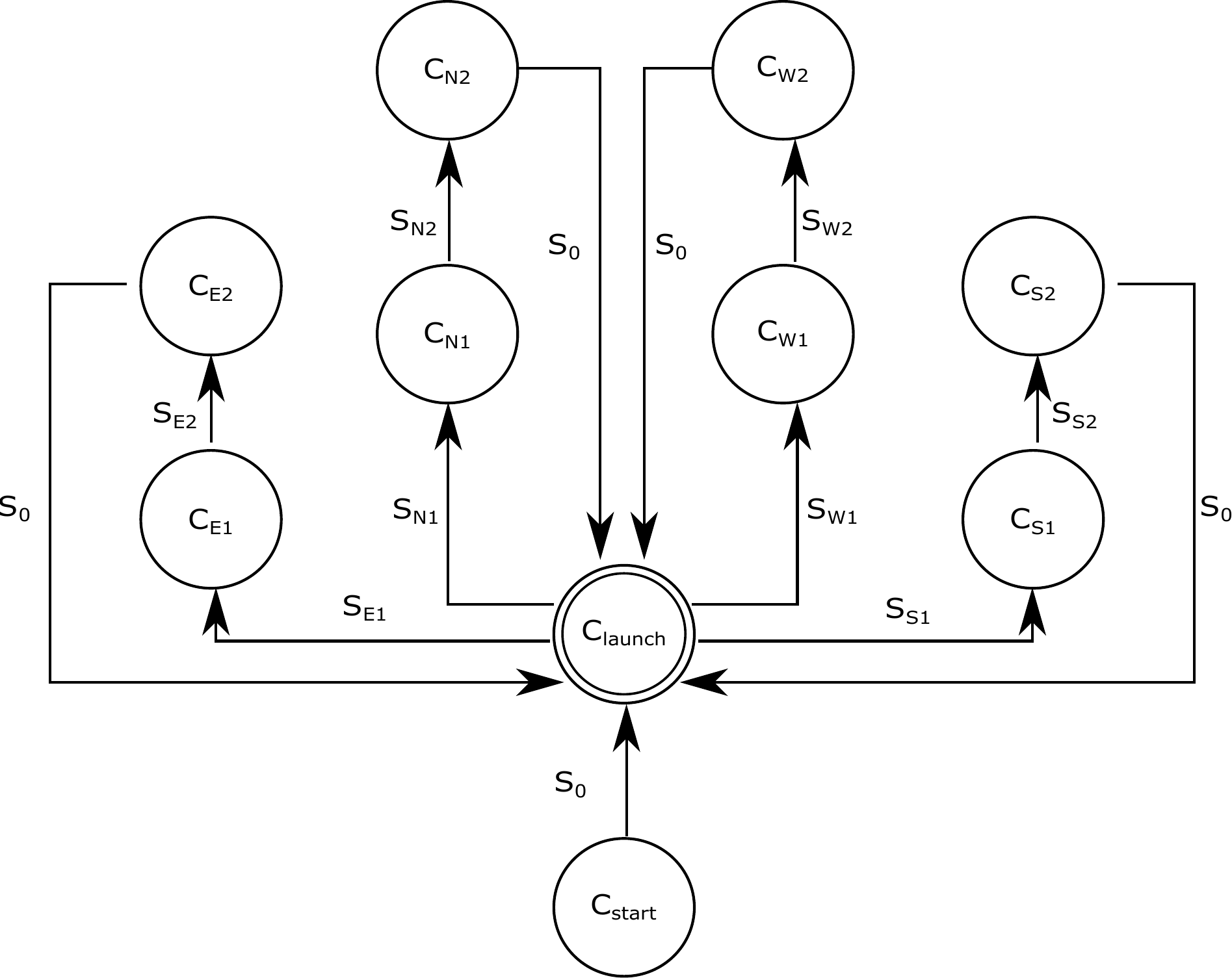}
       		  \caption{Flowchart}
	        \label{fig:Automaton}
  	\end{subfigure}
  		\hspace{.2cm}
	\begin{subfigure}[b]{.48\textwidth}
	        \includegraphics[width=1.0\textwidth]{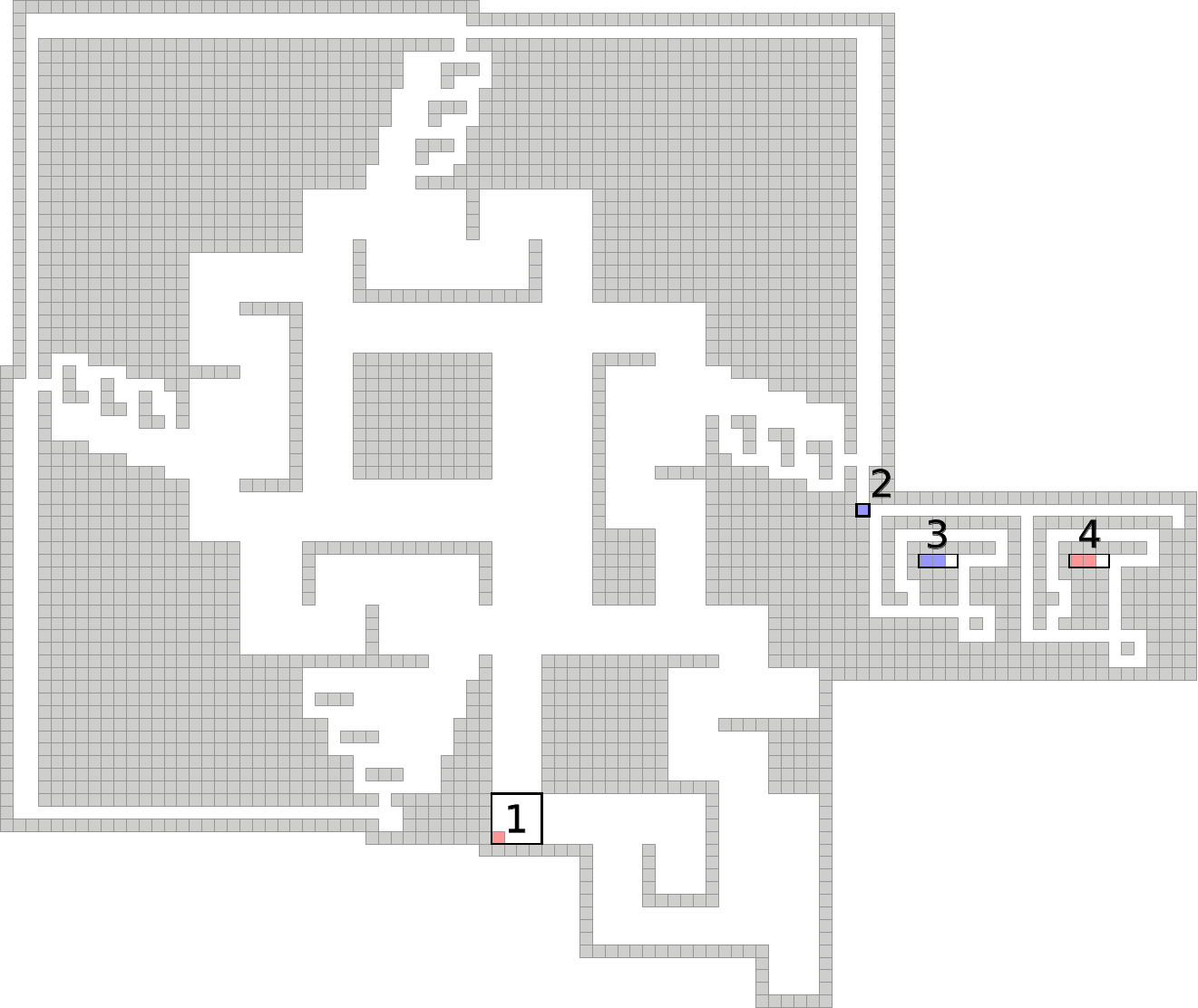}
	        \caption{Example $C_{\text{launch}}$ Position}
	        \label{fig:Launch}
	  \end{subfigure}
    \caption{(a) A flowchart where each state represents a set of configurations and the symbols  represent sequences that can move from one state to another. The sequences for each of the symbols is shown in Table \ref{tab:droptilts}. (b) An example configuration in the $C_{\text{launch}}$ state. Configurations in $C_{\text{launch}}$ always have the assembly located in box 1 at the rightmost bottom corner, the next fuel tile to be shot located in box 2, and all fuel pieces in the fuel chamber are in their proper reservoir pushed to the far left as depicted in box 3 and 4. A configuration \textit{strongly} representing a drop shape $u \in U$ is in  $C_{\text{launch}} $ with no more tiles to launch and the fuel chamber empty.}
    \label{fig:shuriken_proof}
\end{figure*}

\subsection{Universal Drop Shape Builder: Theorem and Proof}

\begin{theorem}\label{thm:shuriken}
	Given two positive integers $h,w \in \mathbb{Z}^+$, there exists a configuration $C$ which is \emph{strongly universal} for the set of drop shapes $U = \{ u | u \subseteq \{ 1,\dots,h\} \times \{ 1,\dots,w\} \}$. WLOG, let $h \geq w$. This configuration has size $\mathcal{O}(h^2w)$ and uses $\mathcal{O}(h^2w)$ tilts to reconfigure into a configuration which strongly represents any shape $u \in U$.
\end{theorem}

\begin{proof}
	This is a proof by construction. We begin with a configuration $C$ where all chambers are empty except the fuel chambers. Following the process outlined in Figure \ref{fig:Automaton}, we can extract the desired fuel piece, and the fuel piece can be moved from the alignment chamber to the shape via the construction chamber. We repeat this process to add single fuel pieces to the shape from any direction. After the fuel chambers are empty, we have generated a configuration $C'$ from $C$ by performing a series of tilts. Thus, $C \rightarrow_* C'$, and to show the transition from the starting configuration to the configuration that represents shape $u \in U$, we use the process shown in the flowchart of Figure \ref{fig:Automaton}.
\end{proof}

\section{PSPACE-complete Results}

This section is a continuation of the work done by \cite{BecDemFek2014PCRA, BecDemFek2014RMPS,BecDemFek2017PCCAL}. In their papers, they discuss the occupancy problem in a tilt-based system which asks whether a given location can be occupied by any tile on the board. We extend their line of investigation by introducing a set of new problems based on their work. We then go on to show PSPACE-completeness for the problems defined.

This section is divided into 5 main subsections. In Section~\ref{subsec:tilt_probs}, we present several natural decision problems for the tilt model. Section~\ref{subsec:puzzle} is where we summarize the puzzle solvability problem from \cite{DBLP:conf/fun/DemaineGLR18}, which we use for our PSPACE-hardness reductions. In Section~\ref{subsec:relocation}, we show that the relocation problem is PSPACE-complete, even when limiting all movable particles to $1\times 1$'s  with a single $2\times 2$. In Section~\ref{subsec:reconfiguration} we show that the reconfiguration problem is PSPACE-complete, even with this same limitation on the size of the movable particles. Finally, Section~\ref{limited_geometry} presents a different construction which also shows PSPACE-completeness for the relocation problem when using a combination of $1\times 1$, $1\times 2$, and $2\times 1$ particles, but using only rectangular board geometry. This same construction is used to show that the occupancy problem is PSPACE-complete as well.

\subsection{Problems in the Tilt Model}\label{subsec:tilt_probs}
In this section we define the three tilt decision problems we consider.

\paragraph{Occupancy Problem} The occupancy problem asks whether or not a given location can be occupied by any tile on the board. Formally, given a configuration $C=(B,P)$ and a coordinate $e \in B$, does there exist a tilt sequence such that $C \rightarrow_* C'$ where $C' = (B,P')$ and $\exists p \in P'$ that contains a tile with coordinate $e$?

\paragraph{Relocation} The relocation problem asks whether a specified polyomino can be relocated to a particular position. That is, given a configuration, a polyomino within that configuration, and a translation of that polyomino, does there exist a sequence of tilts which moves the original polyomino to its translation?

\paragraph{Reconfiguration} The reconfiguration problem asks whether a configuration can be reconfigured into another. Formally, given two configurations $C = (B,P)$ and $C' = (B,P')$, does there exist a tilt sequence such that $C \rightarrow_* C'$?

The occupancy problem was shown to be NP-Hard by \cite{BecDemFek2014PCRA, BecDemFek2014RMPS,BecDemFek2017PCCAL} when only using $1\times 1$ tiles. Their reduction also works to show NP-hardness for the relocation problem when only using $1\times 1$ tiles. Additionally, they show that finding the \emph{optimal} tilt sequence to reconfigure a board from one configuration to another is PSPACE-complete. In this section, we show that the relocation, reconfiguration, and occupancy problems are all PSPACE-complete.
Our hardness proofs rely on a recent result in \cite{DBLP:conf/fun/DemaineGLR18}. They indroduce the \emph{puzzle solvability} problem, which asks whether or not a robot can traverse a system of specific types of gadgets, and show that it is PSPACE-Complete. We present a summary of this problem, and then show our reductions.

\subsection{Puzzle Solvability}\label{subsec:puzzle}
The model introduced in \cite{DBLP:conf/fun/DemaineGLR18} is based on single-agent robot motion planning problems, where an agent navigates a given environment to a target destination. The authors introduce a general model of gadgets, of which the environment to be traversed may be comprised.
We define several key components of these gadgets, as well as the specific types of gadgets themselves. We then state the puzzle solvability problem, which we later reduce from.

Gadget Components:
\begin{itemize}
    \item{Locations.} A gadget consists of one or more \emph{locations}, which are the points of entry and exit to the gadget.

    \item{States.}  Each \emph{state} $s$ of the gadget defines a labeled directed graph on the locations, where a directed edge $(a, b)$ with label $s'$ means that the robot can enter the gadget at location $a$ and exit at location $b$, forcing the state to change to $s'$. The gadgets are defined by state spaces.  A \emph{state space} is a directed graph whose vertices are state and location pairs, where a directed edge from $(s, a)$ to $(s', b)$ means that the robot can traverse through the gadget from $a$ to $b$ if it is in state $s$, such that traversing will change the state of the gadget to $s'$. We use gadgets with at most two states.

    \item{Toggle.}  A \emph{toggle} is a tunnel that can only be traversed in a single direction as dictated by its state. Each toggle has 2 states dictating which direction a robot is allowed to traverse the tunnel. Tunnels are routes between two locations that the robot can traverse through. The state of the toggle gadget changes when the tunnel is traversed.

    \item{Lock.} A \emph{lock} is a tunnel which has two states, locked and unlocked. The unlocked state allows bidirectional traversal. The locked state does not allow traversal whatsoever.
\end{itemize}

Gadgets and Puzzles:
\begin{itemize}
  \item{C2T.}  \emph{Crossing 2-Toggle} is a gadget that has two toggle tunnels perpendicular to each other. Traversing either tunnel causes the gadget's state to change, which means the state (or direction) of both tunnels are changed (or reversed). Figure \ref{fig:C2T} shows a representation of both states of the C2T gadget.

  \item{CTL.} The \emph{Crossing Toggle-Lock} gadget also has two perpendicular tunnels. One tunnel is a toggle, and the other is a lock. Traversing the toggle switches the lock between its locked and unlocked state. Figure~\ref{fig:CTL} shows the representation of both states of the CTL gadget.

  \item{Puzzle.}  A \emph{puzzle} is a problem posed as a system of interconnected gadgets, their initial states, the wires connecting them, and the robot's start and goal location. A puzzle is said to be solvable if there is a path from the start location to the goal location using only moves allowed by the wires and gadgets.
\end{itemize}

\paragraph{Puzzle Solvability Problem} This problem simply asks if a given puzzle of gadgets is solvable. In other words, does there exists a sequence of moves that relocates the robot from its start location to its goal location? If the puzzle consists of only C2T gadgets, we refer to this as the C2T puzzle solvability problem. The same goes for the CTL gadget.

\begin{theorem}\label{puzzle_hardness}
  The C2T puzzle solvability problem and the CTL puzzle solvability problem are both PSPACE-complete.
\end{theorem}

\begin{proof}
  This was shown in Corollary 5.2 from \cite{DBLP:conf/fun/DemaineGLR18}.
\end{proof}

\begin{figure}[t!]
  \centering
    \begin{subfigure}[t]{0.49\textwidth}
      \centering
      \includegraphics[width=0.5\textwidth]{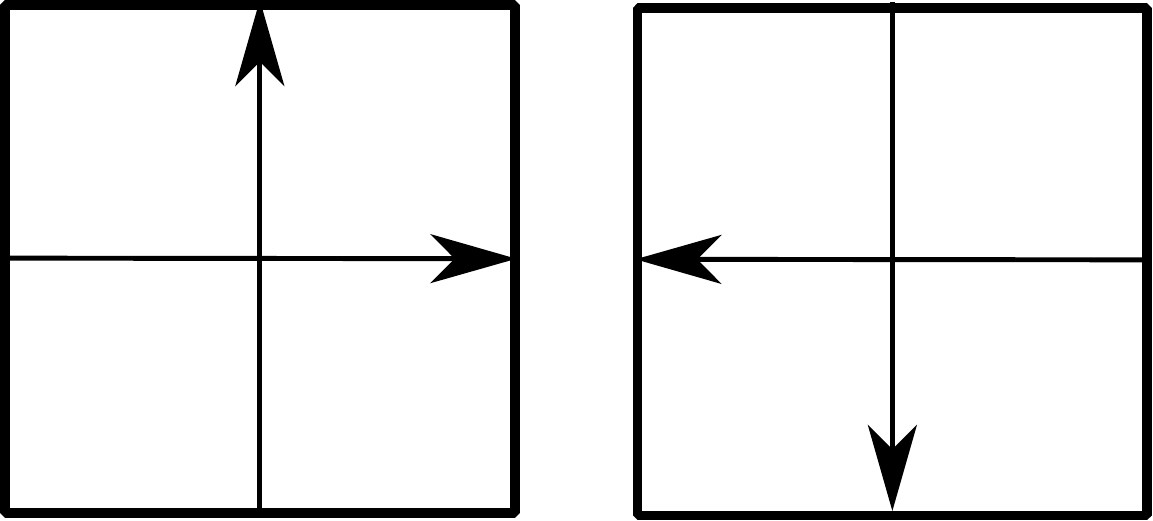}
      \caption{Crossing 2-Toggle (C2T)}
      \label{fig:C2T}
    \end{subfigure}
    \begin{subfigure}[t]{0.49\textwidth}
      \centering
      \includegraphics[width=0.5\textwidth]{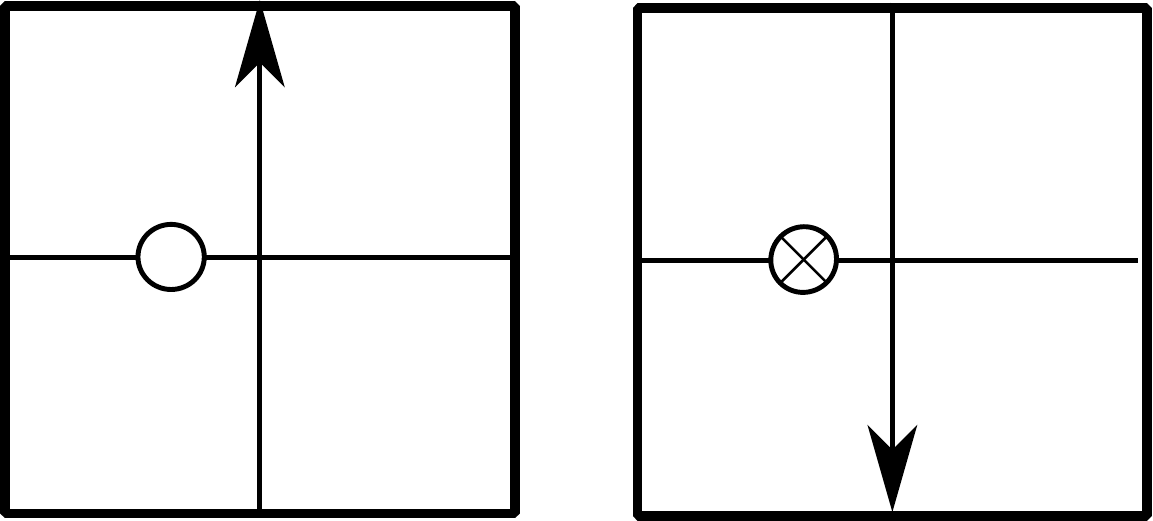}
      \caption{Crossing Toggle-Lock (CTL)}
      \label{fig:CTL}
    \end{subfigure}
    \caption{
    (a) The two states of the crossing 2-toggle gadget. Robot traversal through either tunnel toggles the gadget between these two states. (b) The two states of the crossing toggle-lock gadget. Robot traversal through the directed tunnel toggles the gadget between these two states (locked and unlocked).
    }
    \label{fig:gadget_stuff}
\end{figure}

\subsection{Relocation} \label{subsec:relocation}
The work from \cite{BecDemFek2014PCRA} implies that the relocation problem is NP-hard when using only $1\times 1$ movable particles. In this section, we show that the problem is PSPACE-complete when adding a single $2\times 2$ polyomino.


\subsubsection{Relocation Gadget Preliminaries}

Here we present a gadget in the tilt model that behaves like the C2T gadget described above. Figure \ref{fig:reloc_overview} shows an overview of the gadget. Figure \ref{fig:reloc_sec} shows the basic sections of the gadget which we describe below.
The main idea is that a $1\times 1$ tile (referred to as the \emph{state tile}) is confined to one of two different paths within the gadget and a $2\times 2$ polyomino (referred to as the \emph{robot polyomino}) may only traverse the gadget in a specified direction if the state tile is in the correct corresponding path.
To implement this dependency, we created locations within the gadget where the state tile and robot polyomino need to use each other's geometry to traverse the gadget.
Figures \ref{fig:reloc_all} and \ref{fig:reloc_robot} show the paths of the state tile and robot polyomino, respectively. We note some fundamental properties of the C2T gadget which we recreate in the full-tilt model. Figure \ref{fig:reloc_hl} shows an example of a pathway not reachable by the robot polyomino without the geometric assistance of a state tile. This is the same as the C2T gadget not allowing the robot's traversal through an exit location. In Figure \ref{fig:reloc_hl} we depict an example of a state tile confined to a pathway that is moved into another with the assistance of the robot polyomino. This depicts one of two different \emph{states} that the state tile can be in. This mimics the C2T's state changing function. By combining these elements we are able to create a complex gadget which is able to allow a robot to traverse through it only if the gadget is in the correct state.


\begin{figure*}
    \centering
    \begin{subfigure}[b]{0.32\textwidth}
          \centering
        \includegraphics[width=1\textwidth]{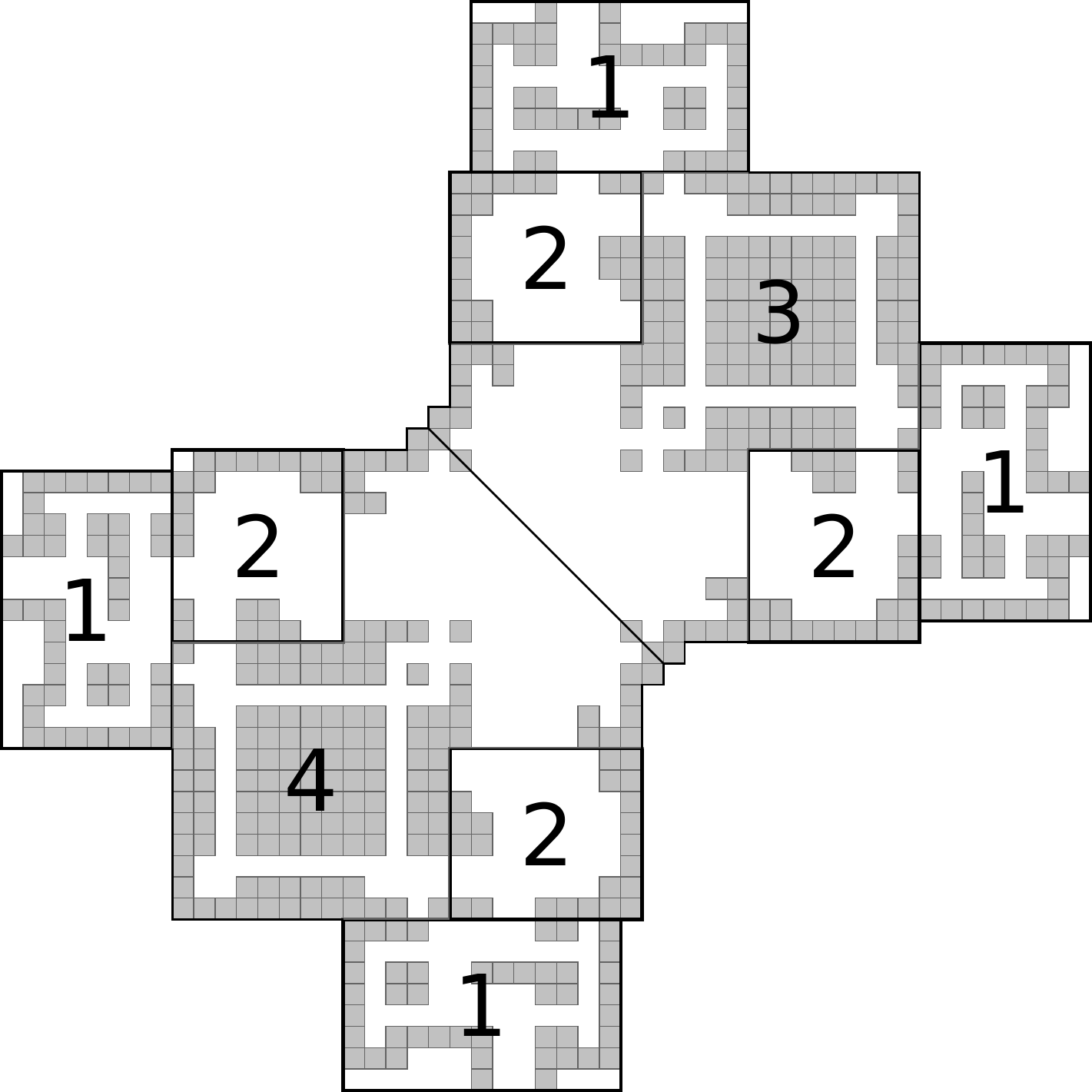}
        \caption{Sections of Gadgets}
        \label{fig:reloc_sec}
    \end{subfigure}
    \hspace{.1cm}
    \begin{subfigure}[b]{0.32\textwidth}
           \centering
         \includegraphics[width=1\textwidth]{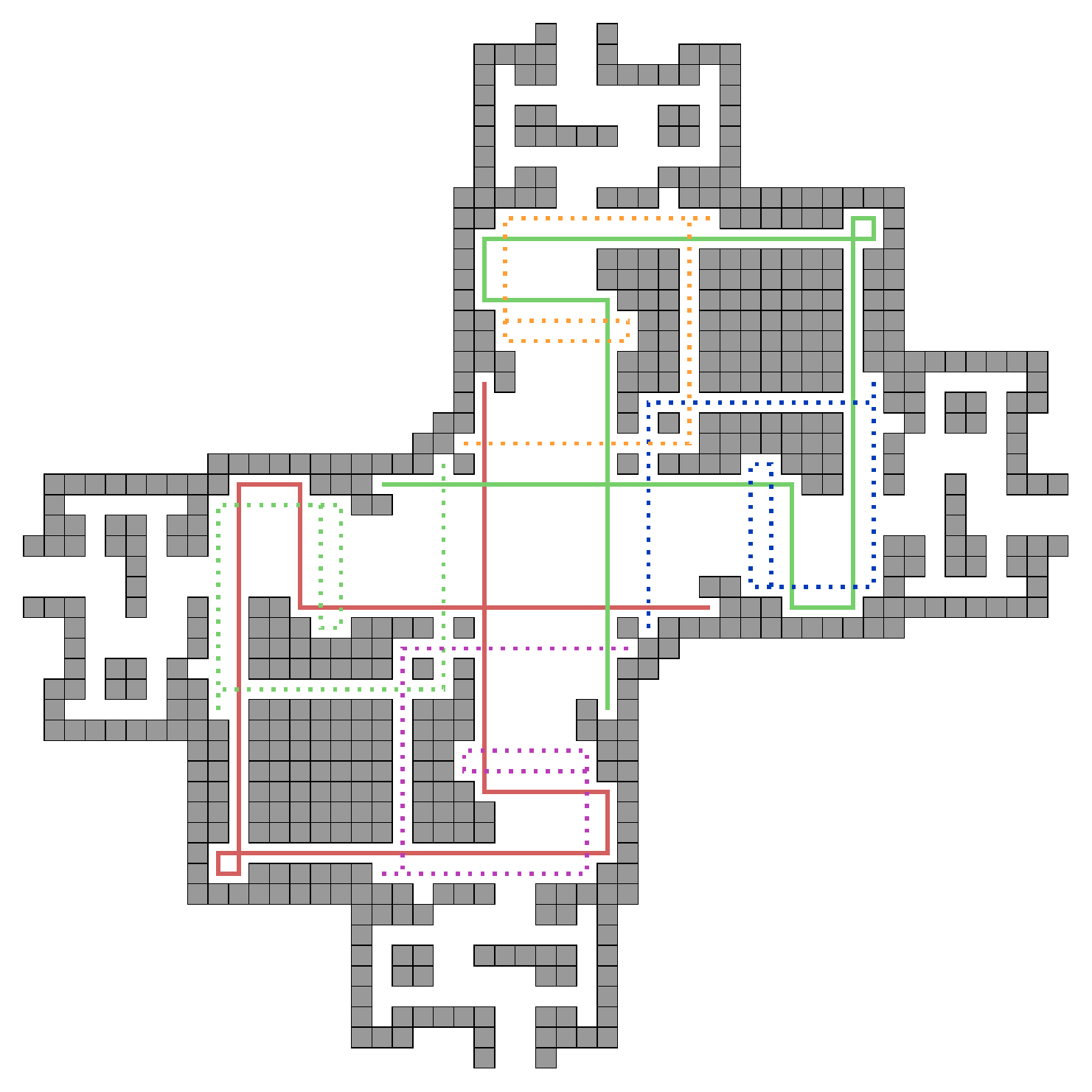}
         \caption{Paths of the State Tile}
         \label{fig:reloc_all}
    \end{subfigure}
    \hspace{.1cm}
    \begin{subfigure}[b]{0.32\textwidth}
           \centering
         \includegraphics[width=1\textwidth]{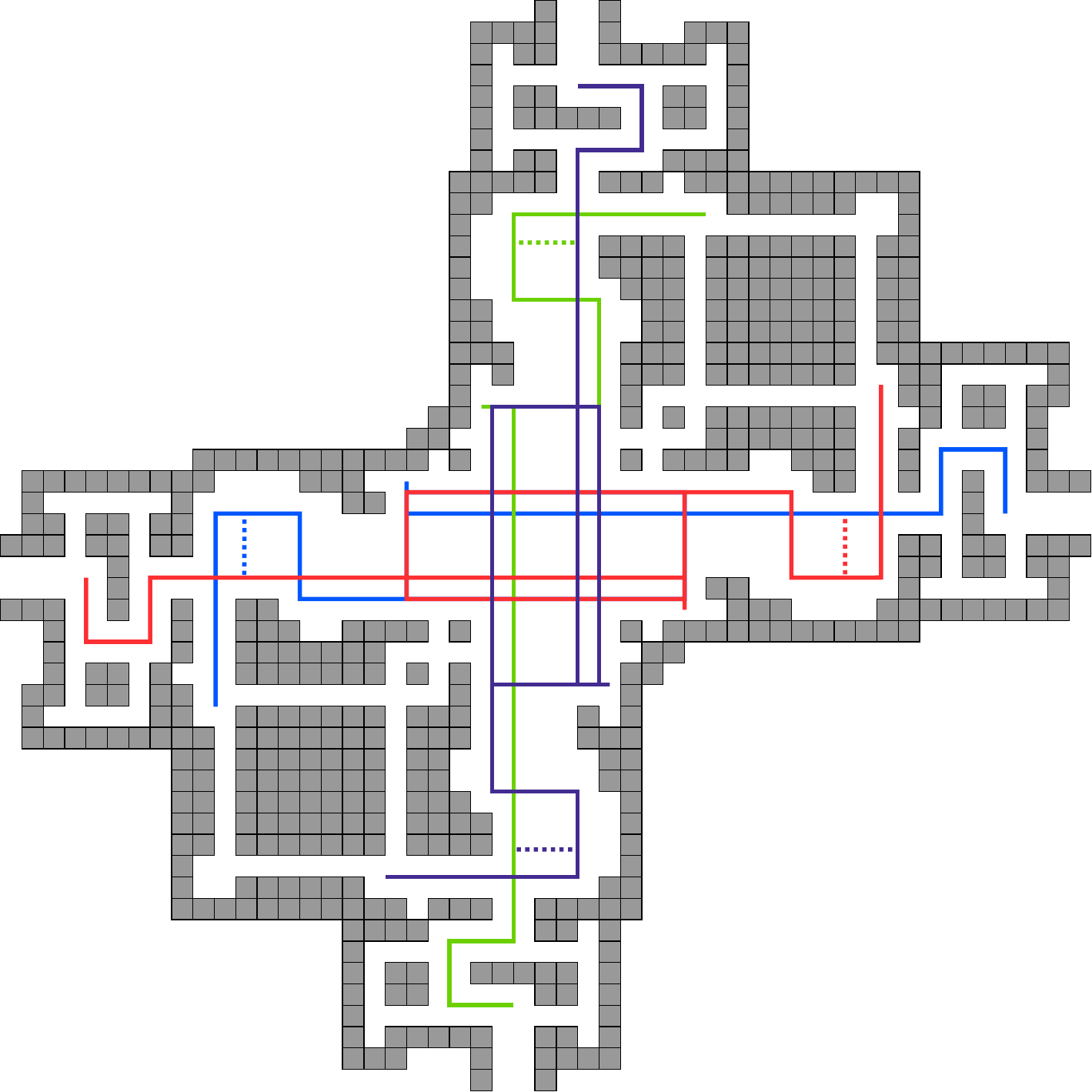}
         \caption{Paths of the Robot Polyomino}
         \label{fig:reloc_robot}
    \end{subfigure}
    \caption{(a) Relocation sections where 1) represents entrance/exits locations, 2) represents areas where the robot polyomino becomes stuck if unassisted by the state tile, and 3) and 4) represent the $NE$ and $SE$ state tile areas. State and robot paths are shown in (b) and (c), where the state tile is stuck on the solid lines and traverses through the dotted lines when changing states, and where the robot polyomino travels through one location to the next through the solid lines, unassisted by the state tile, or traverses the dotted lines if assisted by the state tile.}
    \label{fig:reloc_overview}
\end{figure*}

\begin{figure}[t]
\begin{subfigure}[b]{0.23\textwidth}
  \centering
       \includegraphics[width=1\textwidth]{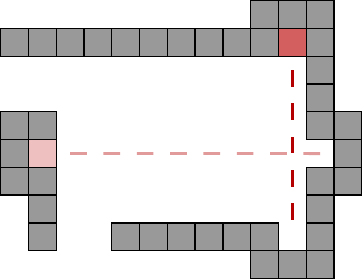}
       \caption{State Path Depiction}
       \label{fig:reloc_hl1}
   \end{subfigure}
   \begin{subfigure}[b]{0.23\textwidth}
   \centering
       \includegraphics[width=1.\textwidth]{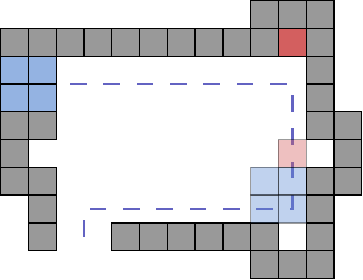}
       \caption{Enforced State Change}
       \label{fig:reloc_hl2}
   \end{subfigure}
   \begin{subfigure}[b]{0.23\textwidth}
     \centering
       \includegraphics[width=1.\textwidth]{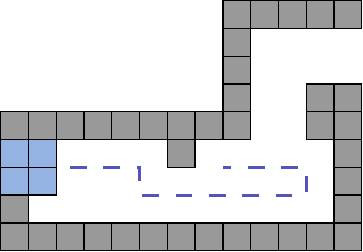}
       \caption{Robot Limited Path}
       \label{fig:reloc_hl3}
   \end{subfigure}
   \begin{subfigure}[b]{0.23\textwidth}
      \centering
        \includegraphics[width=1\textwidth]{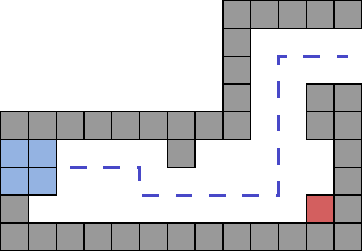}
        \caption{Robot Assisted Path}
        \label{fig:reloc_hl4}
   \end{subfigure}
   \caption{Relocation Properties. (a) A gadget with states 1 (dark) and 2 (light). (b) The robot-traversal ``toggles'' the gadget. (c) A gadget cannot be traversed by the robot alone. (d) The same gadget being robot-traversed with the help of the state tile.}
   \label{fig:reloc_hl}
\end{figure}

\paragraph{Robot Polyomino} The \emph{robot polyomino} is a $2\times2$ polyomino that traverses through a puzzle.
Figure \ref{fig:reloc_robot} shows all paths of the robot in a gadget. The dotted lines show its path without the help of the state tile.

\paragraph{State Tile} The relocation gadget has a \emph{state tile}, which is a single tile trapped inside the gadget. The state tile can freely move in one of the solid lines shown in Figure \ref{fig:reloc_all}, where the dotted lines depict a pathway traversable by the state tiles only with the assistance of the traversing robot polyomino's geometry.

\paragraph{Relocation Gadget}
Our Relocation gadget consists of openings at its four cardinal directions which we will call $N$,$E$,$S$,$W$. The internal geometry is diagonally symmetrical and allows for traversal between its openings. Wires are made with 2-tile wide hallways attached at gadget openings. A high-level overview of the four relocation gadget sections is shown in Figure~\ref{fig:reloc_hl}.

\paragraph{Entrance/Exit Chambers}
Marked as section 1 in \ref{fig:reloc_sec}, the \emph{entrance/exit chambers} represent \emph{locations} in a C2T gadget. These chambers do not enforce any behavior or move sequence constraints on a robot polyomino, letting it move in or out of the chamber with no difficulty. However, these chambers have spaces on its sides designed to keep the state tile within the gadget. Once a state tile becomes stuck in the spaces, the gadget becomes inoperable. These chambers change from entrance to exit chambers depending on the state of the gadget, where if the state tile is in the $NE$ side of the gadget, the $NE$ chambers become entrance chambers and the $SW$ side chambers become exit chambers and vice versa. These different states correspond to the depicted states of C2T gadgets in Figure \ref{fig:gadget_stuff}.

\paragraph{Assistance Chamber}
Marked as section 2 in \ref{fig:reloc_sec}, an \emph{assistance chamber} is where the state tile meets the robot polyomino to assist it in its traversal through the gadget. The robot polyomino can reach the assistance chamber opposite of where it entered from. The only way a robot polyomino can move through an assistance chamber is if the state tile is in the correct path. See Figure \ref{Example_Sequence} for an example travesal. If a robot polyomino enters a gadget whose state tile is in the opposite side, the robot polyomino becomes stuck inside the gadget. This forces the robot to enter through the correct entry points.

\paragraph{State Chambers}
Sections 3 and 4 in \ref{fig:reloc_sec} are the \emph{state chambers} of the relocation gadget. They store the state tiles and dictate the gadget's state. We refer to these states as the $NE$ state/path and the $SW$ state/path. See Figure \ref{fig:reloc_all} for the possible paths of the state tile in its two states. As shown in the figure, the $NE$($SW$) chamber allows a path for the state tile to move freely between its north (south) side and east (west) side.

\paragraph{Directed Tunnels}
We say a \emph{directed tunnel} $(a,b)$ exists in our relocation gadget if a sequence of tilts exists such that we can relocate our robot polyomino from location $a$ to location $b$. For example, the directed tunnel $(W,E)$ exists if our robot polyomino can travel from the west entrance of our relocation gadget to the east entrance.

\begin{figure*}[ht]
	\begin{subfigure}[b]{0.33\textwidth}
	    \centering
        \includegraphics[width=1\textwidth]{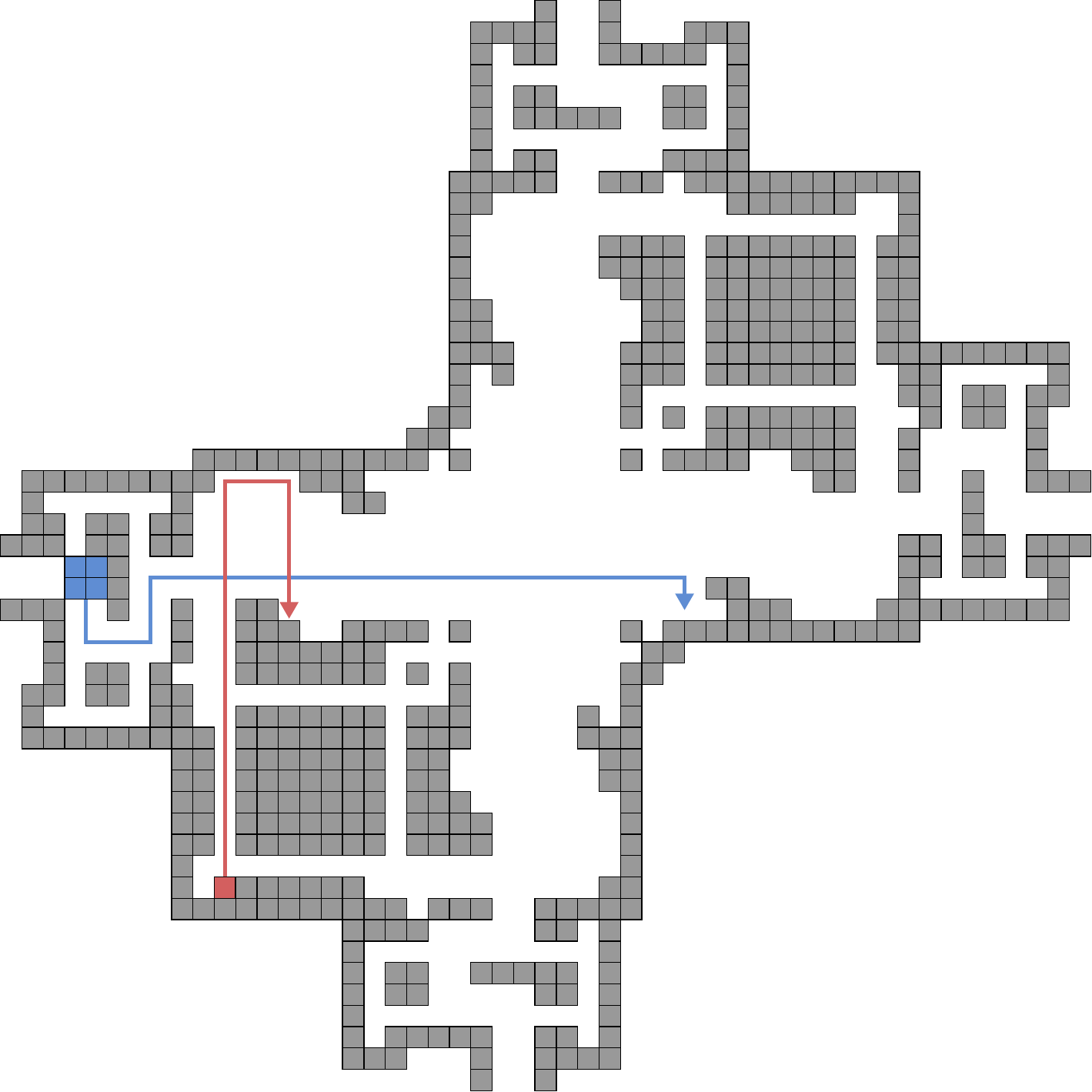}
        \caption{$\langle S, E, N, E, S \rangle$}
        \label{fig:t_ex_1}
    \end{subfigure}
	\begin{subfigure}[b]{0.33\textwidth}
	    \centering
        \includegraphics[width=1\textwidth]{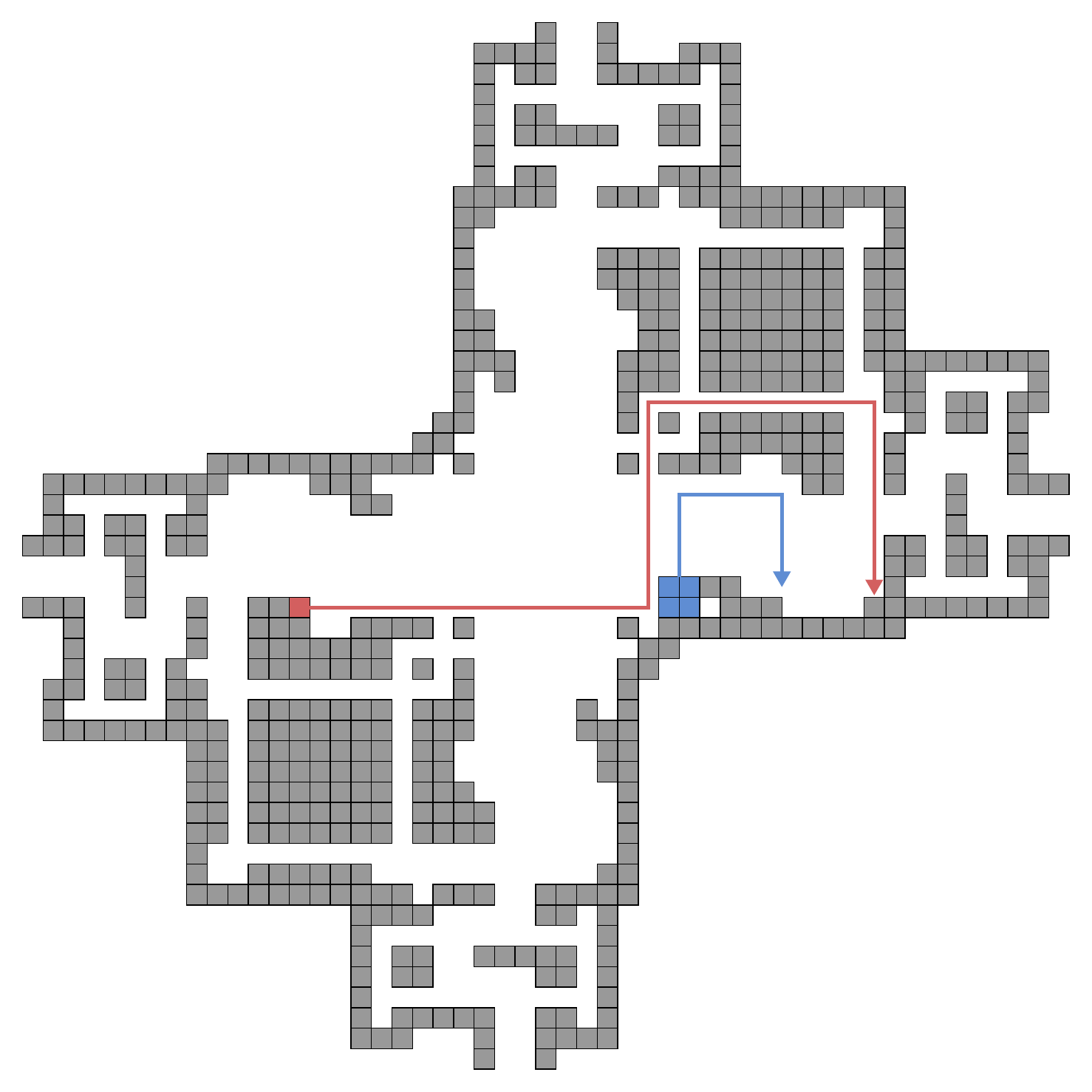}
        \caption{$\langle E, N, E, S \rangle$}
        \label{fig:t_ex_2}
    \end{subfigure}
    	\begin{subfigure}[b]{0.33\textwidth}
	    \centering
        \includegraphics[width=1\textwidth]{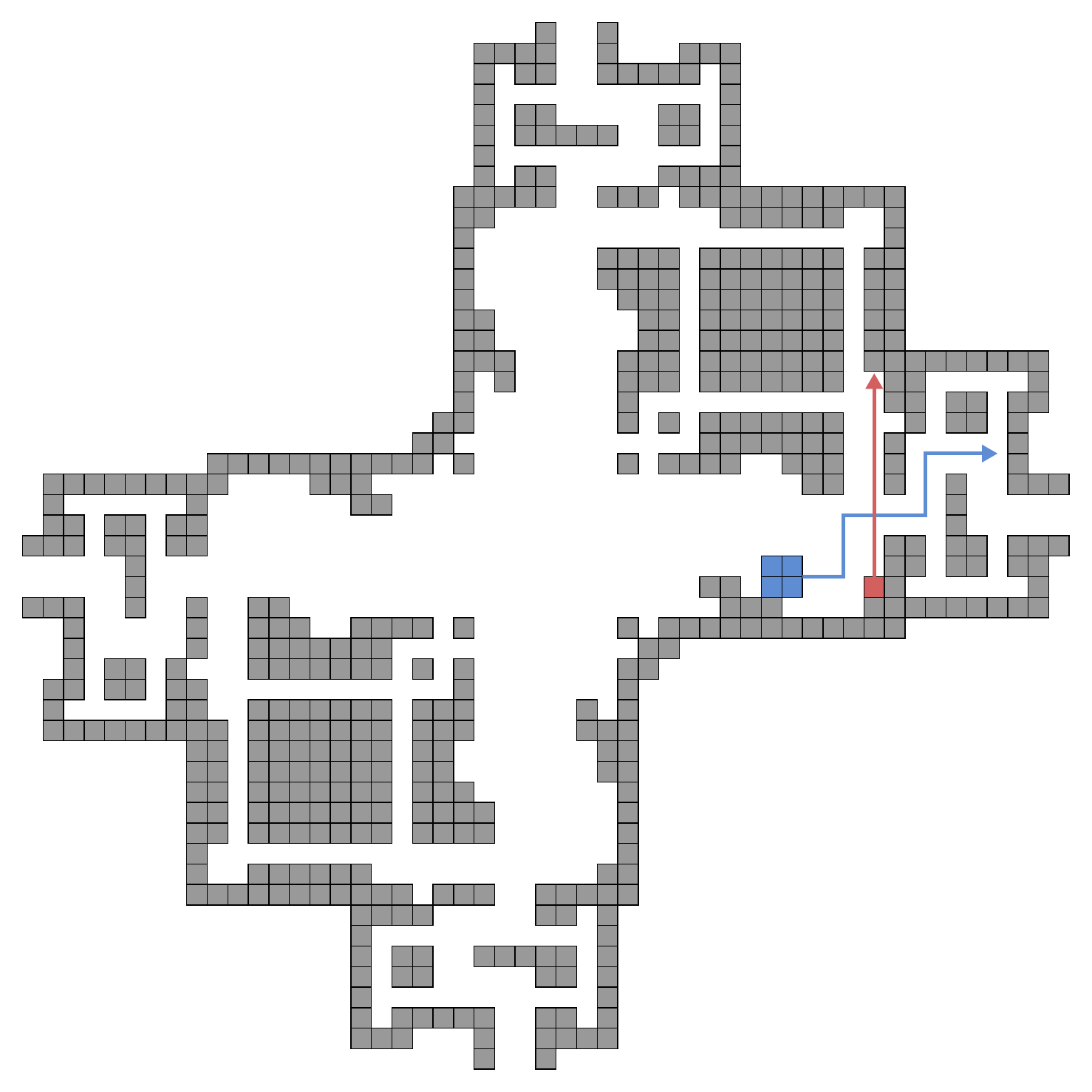}
        \caption{$\langle E, N, E, N, E \rangle$}
        \label{fig:t_ex_3}
    \end{subfigure}
    \caption{Example of a traversing sequence and state change when a gadget is in the state from Figure \ref{fig:C2T}. The robot polyomino enters the gadget in (a) and performs the sequence shown. The robot geometrically assists the state tile to traverse around in (b), and in (c) the state tile assists the robot polyomino via its geometry to exit through the opposite location. In the end the robot would have traversed the gadget, and the state tile would have gone from section 4's red path in Figure \ref{fig:reloc_all} to section 3's green path; Therefore, the gadget was toggled.}
    \label{Example_Sequence}
\end{figure*}

\begin{figure}[ht]
	\begin{subfigure}[b]{0.32\textwidth}
	    \centering
        \includegraphics[width=.5\textwidth]{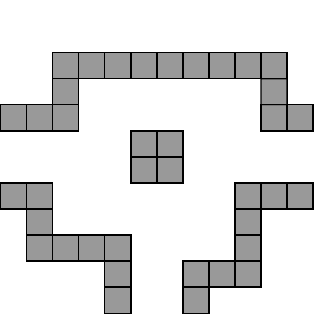}
        \caption{3-way intersection}
        \label{fig:t_shape}
    \end{subfigure}
    \begin{subfigure}[b]{0.32\textwidth}
    	\centering
        \includegraphics[width=.5\textwidth]{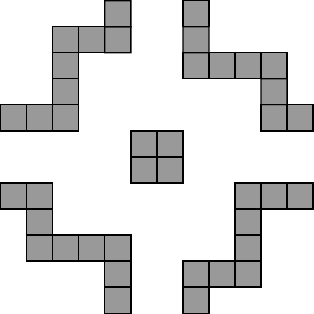}
        \caption{4-way intersection}
        \label{fig:four_way}
    \end{subfigure}
    \begin{subfigure}[b]{0.32\textwidth}
    	\centering
        \includegraphics[width=.6\textwidth]{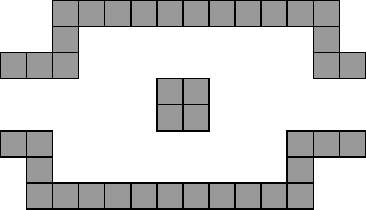}
        \caption{Intermediate Wire}
        \label{fig:inter_arrows}
    \end{subfigure}
    \caption{Intersections of tunnels. The geometric block in the center stops the robot and allows it to choose any of the tunnels to move through. (a) Three tunnels intersecting. (b) A four-way intersection. (c) The intermediate wire used in a system of reconfiguration gadgets to ``reset'' the state tiles.}
    \label{Intersections1}
\end{figure}

\paragraph{Intersections}
To allow robots to change directions in the paths between gadgets, we must have geometry to stop the robot and then choose which direction it will proceed. We place a $2\times2$ blocking piece of geometry in the middle of the wire and expand the surrounding area to allow the robot to make traversing decisions. Examples of 3-way and 4-way intersections are shown in Figures \ref{fig:t_shape} and \ref{fig:four_way}, respectively.

\subsubsection{Relocation is PSPACE-complete}
Here, we show the hardness reduction for the relocation problem. The central idea of our reduction is that our relocation gadget has the same functionality as the C2T gadget. In order to prove this, we prove a series of lemmas (using a brute-force proof-by-computer) about how our gadget works. This equivalent functionality is then used towards the main theorem of this section. We begin with a simple recursive algorithm which, given an initial configuration, creates a tree of all reachable configurations.

\begin{algorithm}
  \caption{Create Configuration Tree}
  \label{alg:config_tree}
  \begin{algorithmic}
    \Function{createTree}{startConfig}
      \State Create a new node for startConfig.
      \State Add startConfig to the list of visited configurations.
        \ForAll {d in \{N,E,S,W\}}
          \State Create new configuration by tiliting startConfig in direction d.
          \If {the new configuration is not in the list of visited configurations}
            \State Create a child node for direction d and call createTree on the new configuration.
          \EndIf
        \EndFor
    \EndFunction
  \end{algorithmic}
\end{algorithm}


\begin{lemma}\label{lem:config_tree}
  Given relocation gadget configurations $C$ and $C'$, there exists a node for $C'$ in the configuration tree created by \emph{createTree}($C$) if and only if  $C \rightarrow_* C'$.
\end{lemma}

\begin{proof}
  Since algorithm~\ref{alg:config_tree} recursively calls itself until the base case where all four child configurations have already been visited, it is easy to see that all reachable configurations are generated by this algorithm.
\end{proof}

\paragraph{Forest Generation} For the following lemmas, we use a brute-force proof-by-computer approach in which we generate a forest of trees (each created by calling algorithm~\ref{alg:config_tree}) from each of the configurations where the robot polyomino is located in any of the four entrance/exit chambers. In other words, a tree is generated for every combination of robot polyomino entrance/exit location and all possible state tile locations. We can easily modify algorithm~\ref{alg:config_tree} to flag nodes as north, east, south, or west nodes if their corresponding configuration has the robot polyomino located at the north, east, south, or west entrance/exit chambers. Sourcecode: \href{https://github.com/asarg/TumbleTiles}{https://github.com/asarg/TumbleTiles}.

\begin{lemma}\label{lem:relocation_tunnels_NE}
  Directed tunnels $(N,S)$ and $(E,W)$ exist in a relocation gadget if and only if the gadget is in the $NE$ state.
\end{lemma}

\begin{proof}
  To prove this, we show that the existance of directed tunnels $(N,S)$ and $(E,W)$ implies that the gadget is in the $NE$ state, and that the gadget being in the $NE$ state implies the existance of directed tunnels $(N,S)$ and $(E,W)$.
  By Lemma~\ref{lem:config_tree}, we see that the $(N,S)$ / $(E,W)$ directed tunnels exist for any north (east) configuration which is the root of a tree that contains a south (west) node (i.e. there is a tilt sequence to move the robot polyomino from the north entrance to the south).
  For the first implication, we search the brute-force forest and verify that all trees which have a north (east) root and contain a south (west) node also begin with the state tile at some location in the $NE$ path.
  For the second implication, we search the brute-force forest and verify that all trees which have a north (east) root and begin with the state tile at some location in the $NE$ path contain a south (west) node.
\end{proof}

\begin{lemma}\label{lem:relocation_tunnels_SW}
  Directed tunnels $(S,N)$ and $(W,E)$ exist in a relocation gadget if and only if the gadget is in the $SW$ state.
\end{lemma}

\begin{proof}
  We perform a similar forest-searching process as Lemma~\ref{lem:relocation_tunnels_NE} and show a similar implication for directed tunnels $(S,N)$ / $(W,E)$ and gadget state $SW$.
\end{proof}

\begin{lemma}\label{lem:relocation_tunnels_turns}
  Directed tunnels $(N,E), (N,W), (S,E), (S,W), (E,N), (E,S), (W,N), \text{and } (W,S)$ do not ever exist in a relocation gadget (regardless of its state).
\end{lemma}

\begin{proof}
  For each directed tunnel $(a,b) \in \{(N,E), (N,W), (S,E), (S,W), (E,N), (E,S), (W,N), (W,S)\}$, we search the forest and verify that no tree exists which has root flagged for direction $a$ and contains a node flagged for direction $b$.
\end{proof}

\begin{lemma}\label{lem:relocation_state}
  If a directed tunnel in a relocation gadget is traversed, the state of the gadget is toggled.
\end{lemma}

\begin{proof}
  From Lemma~\ref{lem:relocation_tunnels_turns}, we can deduce that the only directed tunnels which may exist in a relocation gadget are $(N,S)$, $(E,W)$, $(S,N)$, and $(W,E)$.
  Consider directed tunnel $(N,S)$. We search the forest and find all trees which have a north root and contain a south node. By Lemma~\ref{lem:relocation_tunnels_NE} we know that these trees have root configurations where the state tile is in the $NE$ path.
  We verify that all south nodes in these trees represent configurations in which the state tile is in the $SW$ path.
  We perform equivalent verifications for the remaining directed tunnels.
\end{proof}

\begin{lemma}\label{thm:relocation_c2t}
	The relocation gadget correctly implements the behavior of the C2T gadget.
\end{lemma}

\begin{proof}
    By Lemmas~\ref{lem:relocation_tunnels_NE} and \ref{lem:relocation_tunnels_SW}, we see that robot traversal through our relocation gadget is dictated by what state the gadget is in. Lemma~\ref{lem:relocation_tunnels_turns} shows that this traversal must be in a straight line. Lemma~\ref{lem:relocation_state} shows that robot traversal forces the gadget to change state. Thus, our relocation gadget has the same functionality as the C2T.
\end{proof}

\begin{theorem}\label{thm:relocation}
  The relocation problem is PSPACE-complete.  Moreover, it remains PSPACE-complete when polyominoes are all $1\times1$ squares with a single $2\times2$ square.
\end{theorem}

\begin{proof}
  First, we observe that this problem is in PSPACE. To see this, consider a directed graph $G = (V,E)$ where each vertex is a configuration on the board $B$ and each edge $e_{i,j} = (v_i,v_j)$ connects two vertices if there exists one tilt that reconfigures $v_i$ to $v_j$. Clearly, a nondeterministic search of this graph will yield the answer in polynomial space, implying membership in NPSPACE. Since NPSPACE $=$ PSPACE, we get membership in PSPACE.

  We now show PSPACE-hardness by a reduction from the C2T puzzle solvability problem from \cite{DBLP:conf/fun/DemaineGLR18}.
  Given an instance of a C2T puzzle, use the tools described above to create a tilt model configuration as follows:
    For each C2T gadget in the puzzle, create a relocation gadget such that the initial state of the relocation gadget matches the initial state of the C2T gadget.
    For each straight wire in the puzzle, create a 2-tile wide tunnel in the tilt configuration.
    For each 3/4-way wire intersection in the puzzle, create a 3/4-way tunnel intersection in the tilt configuration.
    For the goal location of the puzzle, create a $2\times 2$ goal chamber in the the tilt configuration.

  Lemmas~\ref{lem:relocation_tunnels_NE} and \ref{lem:relocation_tunnels_SW} show that there exists a way for the robot polyomino to traverse our relocation gadget. This, along with the fact there clearly exists a move sequence to traverse each wire tunnel/intersection, shows that if the given C2T puzzle instance has a solution, then there exists a tilt sequence which relocates the robot polyomino to the goal chamber. Thus, a solution for the puzzle solvability problem implies a solution to the relocation problem.

  Now, consider the case where the given C2T instance is not solvable. Since Lemma~\ref{thm:relocation_c2t} shows that an individual relocation gadget has the same functionality as an individual C2T gadget, the only way that our relocation problem could be solvable is if the functionality of our relocation gadget changed when placed in a system of gadgets. The only way one relocation gadget could possibly influence another is if a state tile could leave one gadget and enter another. After generating the brute-force forest with a modified algorithm~\ref{alg:config_tree} (flagging any configurations in which the state tile exits the reconfiguration gadget) we search the forest and verify that no such configurations exist. Thus, no solution for the C2T puzzle solvability problem implies no solution for the relocation problem. By Theorem~\ref{puzzle_hardness}, we see that the relocation problem is PSPACE-complete.
\end{proof}

\begin{corollary}\label{corr:occupancy}
  The occupancy problem is PSPACE-complete. Moreover, it remains PSPACE-complete when polyominoes are all $1\times1$ squares with a single $2\times2$ square.
\end{corollary}

\begin{proof}
  This follows from the same construction that was used in Theorem~\ref{thm:relocation}. Since we know that no state tile can ever leave a relocation gadget, the only polyomino which could occupy the goal location is the robot polyomino. Thus, the same implications hold for the C2T puzzle solvability problem and the occupancy problem.
\end{proof}

\subsection{Reconfiguration}\label{subsec:reconfiguration}
In \cite{BecDemFek2014PCRA}, it was shown that computing a shortest sequence of tilts required to transform one configuration into another is PSPACE-complete.
In this section, we show that even the problem of determining if a reconfiguration sequence exists between two given configurations is PSPACE-complete.

\subsubsection{Reconfiguration Gadget Model Preliminaries}
For the reconfiguration problem we must provide a unique final configuration, as opposed to just a final location of a single polyomino. The previous gadget is insufficient for this problem as there is generally not a unique final configuration for all successful traversals, and it would not be polynomial time computable even in the cases that it were. The issue is that we would not know the state of each of the gadgets, and thus not know the locations of the state tiles throughout the system. We address this issue by extending the relocation gadget to allow for a final tilt sequence that moves all state tiles of all gadgets (regardless of the state) into one unique position in each gadget, thereby providing a polynomial time computable target configuration for the reduction.

\paragraph{Reconfiguration Gadget} The reconfiguration gadget is a relocation gadget with additional geometry. Each state tile has additional pathways that lead to an inescapable chamber on the perimeter of the gadget. The robot polyomino otherwise uses the same paths as in the relocation gadget.

\paragraph{Reconfiguration Ring} Each reconfiguration gadget has a \emph{reconfiguration ring} located on its perimeter reachable by the state tiles only. These hallways help solve the reconfiguration problem as they convert all gadgets to one global unique configuration. To insert all state tiles to this ring, we can move the state tile through the new geometry (depicted in Figure \ref{fig:reconfig_all2}) into the reconfiguration ring, and render all gadgets inoperable.

\paragraph{Restarting Positions} We define \emph{restarting positions} (depicted by the solid state tiles in Figure \ref{fig:reconfig_all2}) as the locations where we initialize all state tiles of the same state. Applying a global motion signal will ensure that all state tiles (in the same state) will be in the same position at all times.

\paragraph{Intermediate Wire} An \emph{intermediate wire} (Figure \ref{fig:inter_arrows}) is a small chamber gadget placed on wires between every pair of reconfiguration gadgets to ensure the state tiles will not enter the reconfiguration ring. The intermediate gadget gives the robot enough room to make tilts in all directions. After every reconfiguration gadget traversal, the robot can move freely in the intermediate wire and reposition the state tiles to their restarting positions.

\begin{figure*}[t]
\centering
   \begin{subfigure}[b]{0.45\textwidth}
     \centering
       \includegraphics[width=1\textwidth]{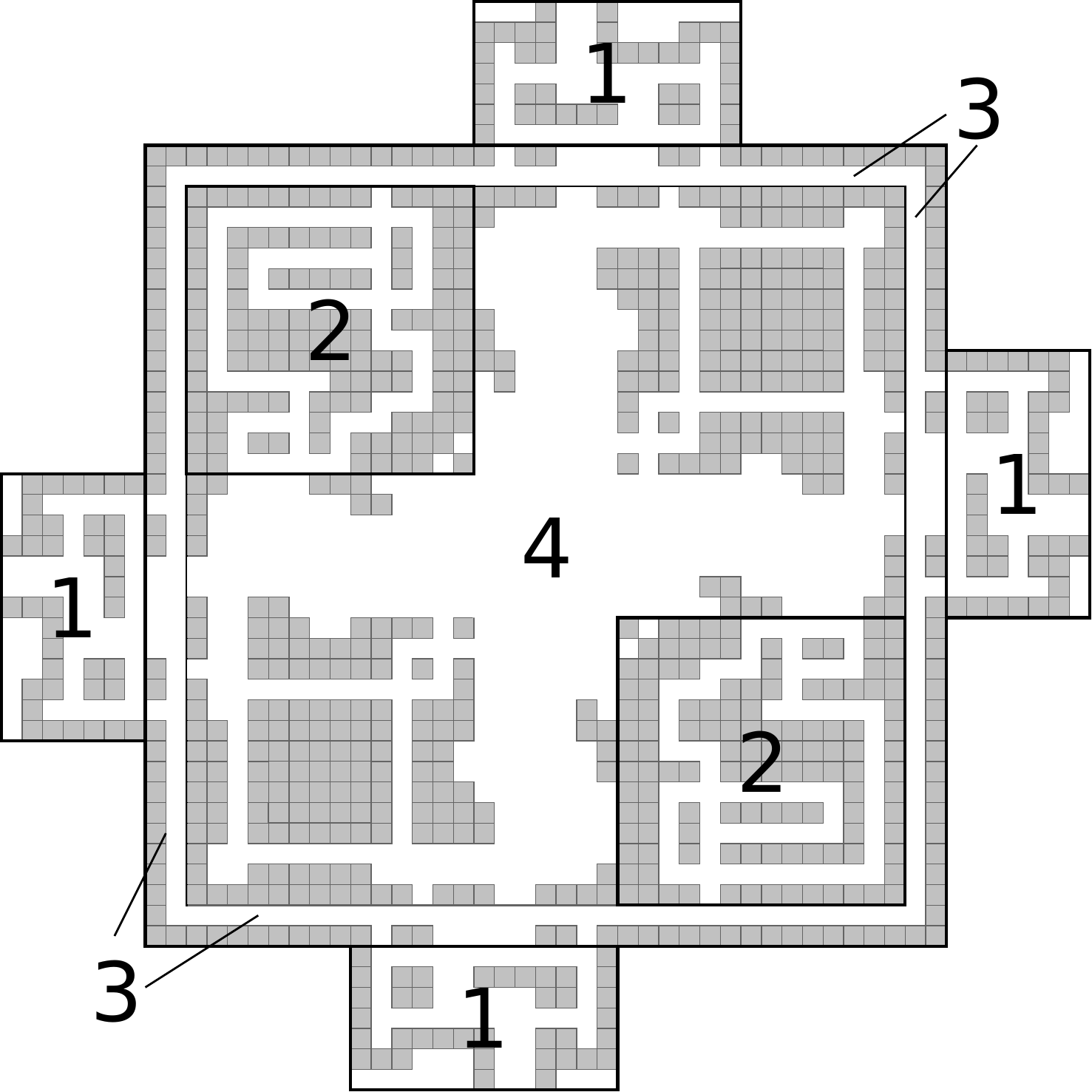}
       \caption{Reconfiguration Gadget Sections}
       \label{fig:reconfig_sec}
   \end{subfigure}
   \hspace{1cm}
   \begin{subfigure}[b]{0.45\textwidth}
      \centering
        \includegraphics[width=1\textwidth]{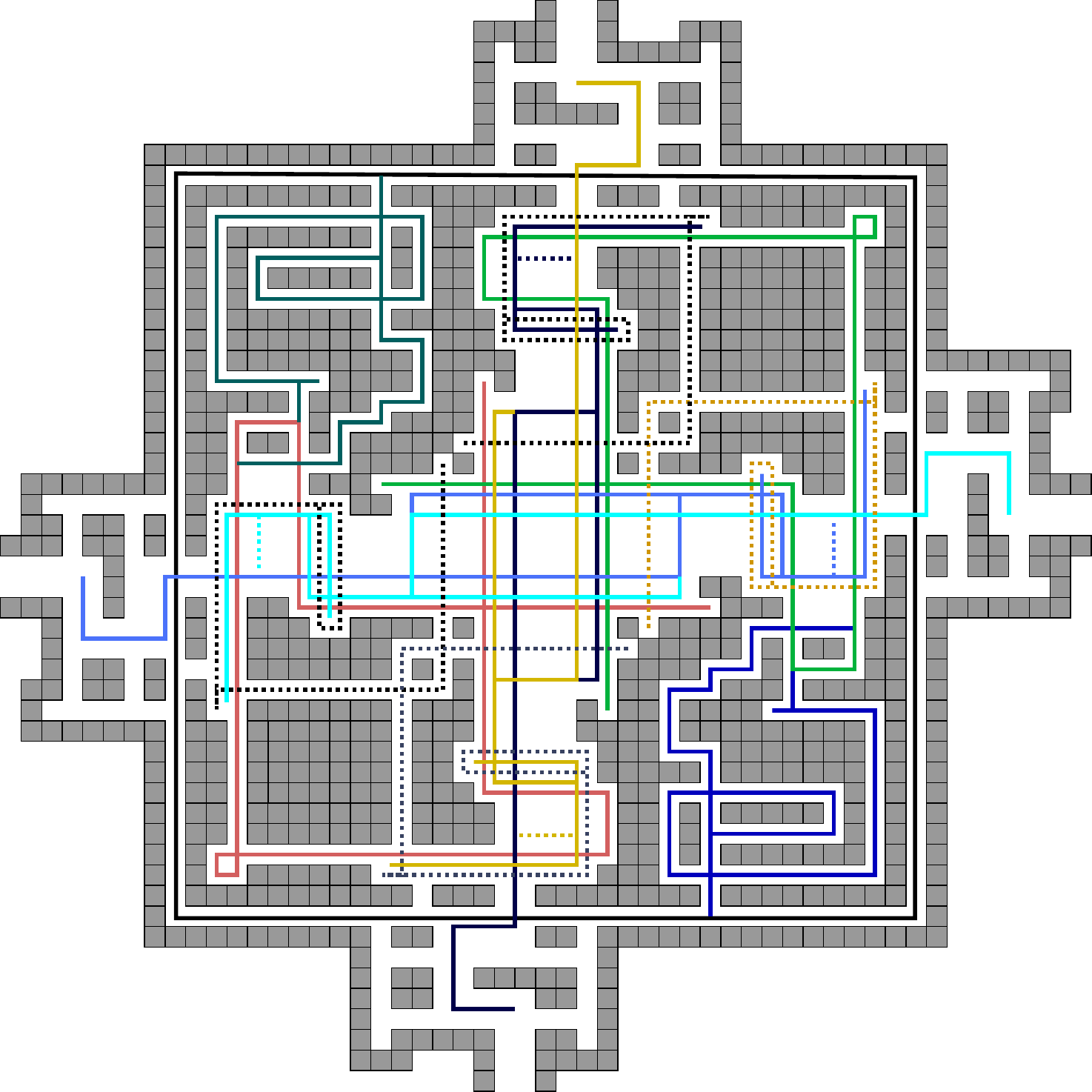}
        \caption{Reconfiguration Gadget Positions}
        \label{fig:reconfig_all}
   \end{subfigure}
   \caption{(a) Sections 1 are the \emph{entrance chambers}, sections 2 are the \emph{pre-reconfiguration tunnels}, section 3 is the \emph{reconfiguration tunnel}, section 4 is the \emph{relocation gadget}. The solid lines depict the pathways the state tile and robot polyomino are free to move through, and the dotted lines depict the pathways which are only accessible through cooperation between the state tile and robot polyomino.}
   \label{fig:reconfig_overview}
\end{figure*}

\subsubsection{Reconfiguration is PSPACE-complete}
\begin{figure*}[t]
\centering
   \begin{subfigure}[b]{0.45\textwidth}
      \centering
        \includegraphics[width=1\textwidth]{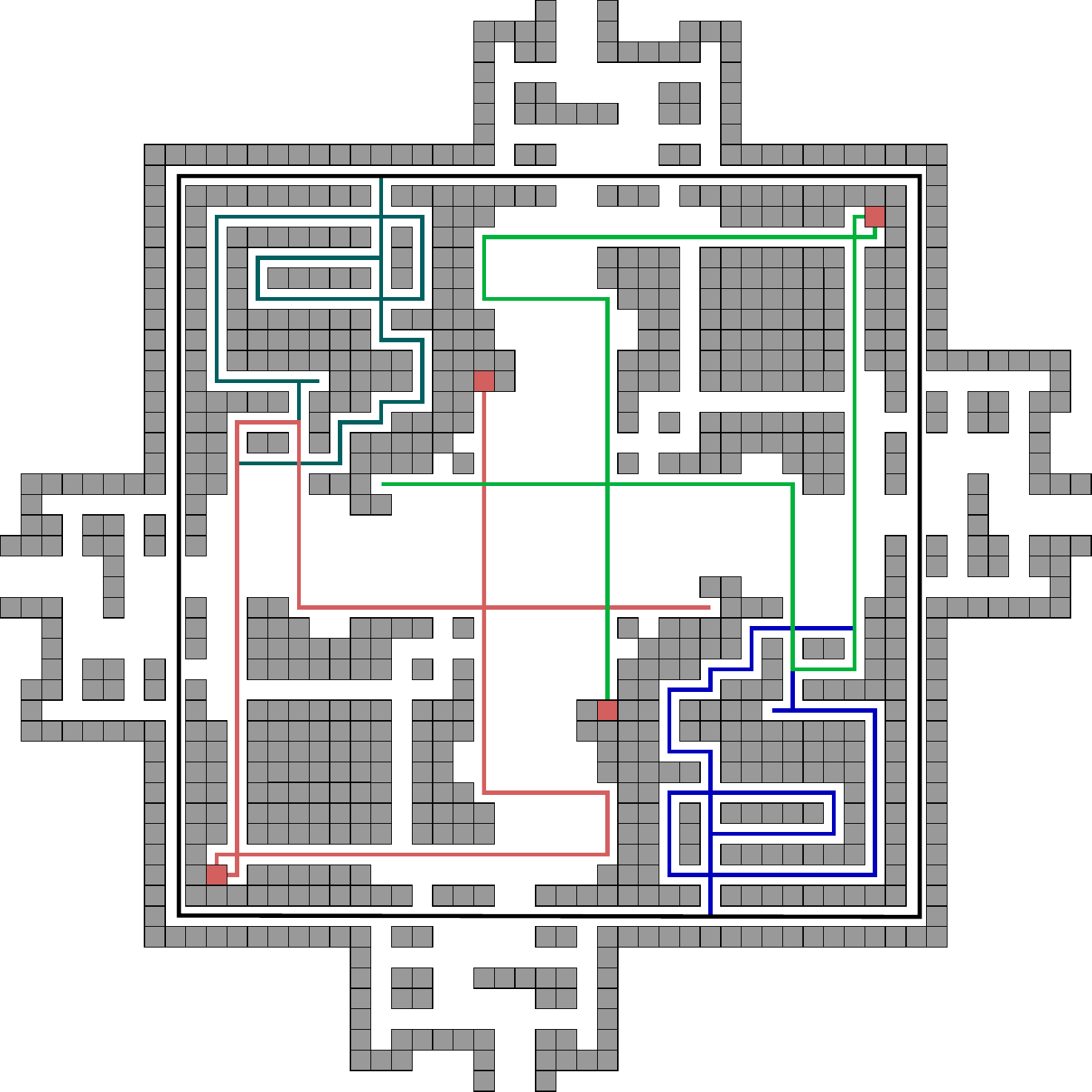}
        \caption{Reconfiguration state tile paths.}
        \label{fig:reconfig_all2}
   \end{subfigure}
   \caption{When the robot has reached its goal location, all the state tiles will be in one of the two state paths (red or light green). The user could then maneuver all state tiles through the paths (dark green and blue) and place them inside the reconfiguration ring. Doing so will convert all gadgets to a global configuration.}
   \label{fig:reconfig_overview2}
\end{figure*}

Using the same proof by computer described above, we show that in a system of reconfiguration gadgets the state tiles can be placed in their restarting positions by using the intermediate wire. We do this to avoid accidentally placing a state tile in the reconfiguration ring, caused by performing tilt sequences that move the robot around the system. We show that by placing intermediate wires before every reconfiguration gadget, we can traverse through one gadget to the next and place the state tiles back in their restarting positions. Therefore, everytime the robot enters a gadget the state tiles will be in one of the two restarting positions we depict in Figure \ref{fig:reconfig_all}. When the robot has been relocated to its appropiate location, we can then proceed to move all state tiles to their reconfiguration ring. This converts all of the gadgets into one configuration.

\begin{lemma}\label{lem:state_aligned}
For every system of gadgets there exists a sequence of tilts that allows the robot to traverse the gadget system without moving the state tile into the reconfiguration ring.
\end{lemma}

\begin{proof}
There exists a set of tilt sequences that allow the robot polyomino to traverse the gadget while preserving the state tile restarting positions defined above. Thus, depending on the direction traveled, every state tile will be in a specific location depending on the gadget's state as the polyomino is exiting the gadget, or trivially be one tilt away. This is true not only for the gadget being traversed, but also for every gadget in the system. This ensures that while traversing through the hallways no state tiles are forced into a reconfiguration ring. From the top-right and bottom-left state tile restarting positions, a total of 13 tilts are required to permanently trap the state tile in the reconfiguration ring, and tilts must be performed in a clockwise order to progress the tile through the reconfiguration hallway. The maximum number of hallways to connect two gadgets or a gadget to an intersection is five. Thus, at most five tilts are needed in a hallway traversal. This means no state tile can enter the reconfiguration ring while only traversing hallways. The intersections are designed to use counter-clockwise tilts in choosing a direction, and therefore does not allow progression through the chamber. Finally, since the entrance to every gadget is preceeded by an intermediate wire, we can perform tilts in these wires to reposition the state tiles to their optimal positions.
In the worst case scenario all state tiles are located in their gadget's reconfiguration hallway, however, we have shown that even from this scenario there is a tilt sequence that will  reposition all state tiles into their optimal starting positions. 
\end{proof}

\begin{theorem}\label{thm:reconfiguration}
  The reconfiguration problem is PSPACE-complete.  Moreover, it remains PSPACE-complete when polyominoes are all $1\times1$ squares with a single $2\times2$ square.
\end{theorem}

\begin{proof}
 Since the reconfiguration gadget shares the same structure and properties as the relocation gadget, it also behaves as a C2T gadget. We use an exhaustive computer simulation, as in the proof of the relocation gadget, to ensure the gadget works as described. Lemma \ref{lem:state_aligned} shows that solving the relocation problem is feasible using just reconfiguration gadgets. Therefore, after solving the relocation problem using reconfiguraton gadgets, we can move all state tiles into their reconfiguration ring. Doing so will move all state tiles into a specific location within each gadget that is specified as part of our final configuration $D$, which also has the required location of the robot.
\end{proof}

\subsection{Relocation and Occupancy with Limited Geometry}\label{limited_geometry}
In this section we show that the relocation and occupancy problems are PSPACE-complete even when limiting the type of geometry allowed in the configuration. We prove this using a rectangular board, single tiles and several $1 \times 2$ and $2 \times 1$ dominoes. To show hardness, we once again reduce from the puzzle solvability problem, but with a system of crossing toggle-lock gadgets rather than crossing 2-toggle gadgets.

\begin{figure}
    \centering
    \begin{subfigure}[b]{0.115\textwidth}
        \centering
        \includegraphics[width = 0.9\textwidth]{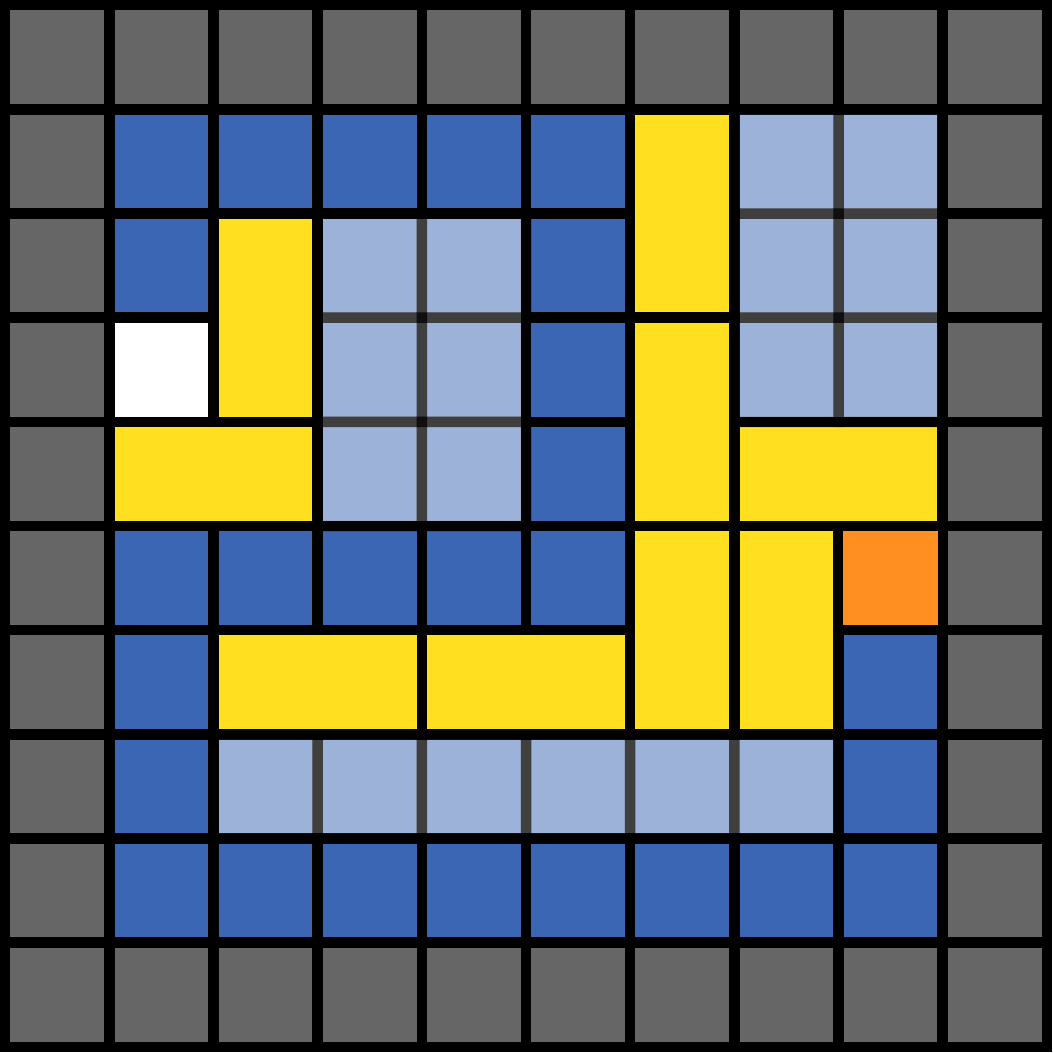}
        \caption{initial}
        \label{fig:full0}
    \end{subfigure}
    \begin{subfigure}[b]{0.115\textwidth}
        \centering
        \includegraphics[width = 0.9\textwidth]{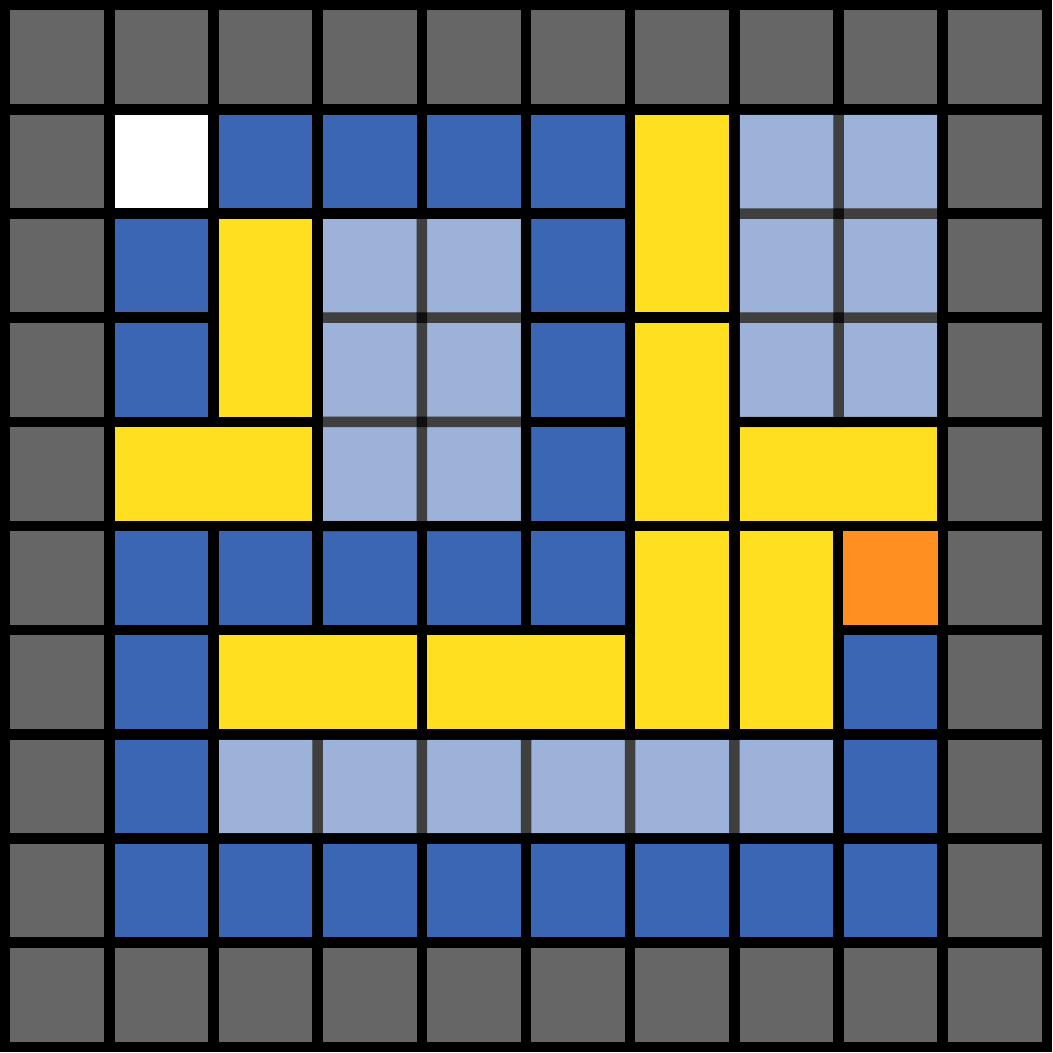}
        \caption{S}
        \label{fig:full1}
    \end{subfigure}
    \begin{subfigure}[b]{0.115\textwidth}
        \centering
        \includegraphics[width = 0.9\textwidth]{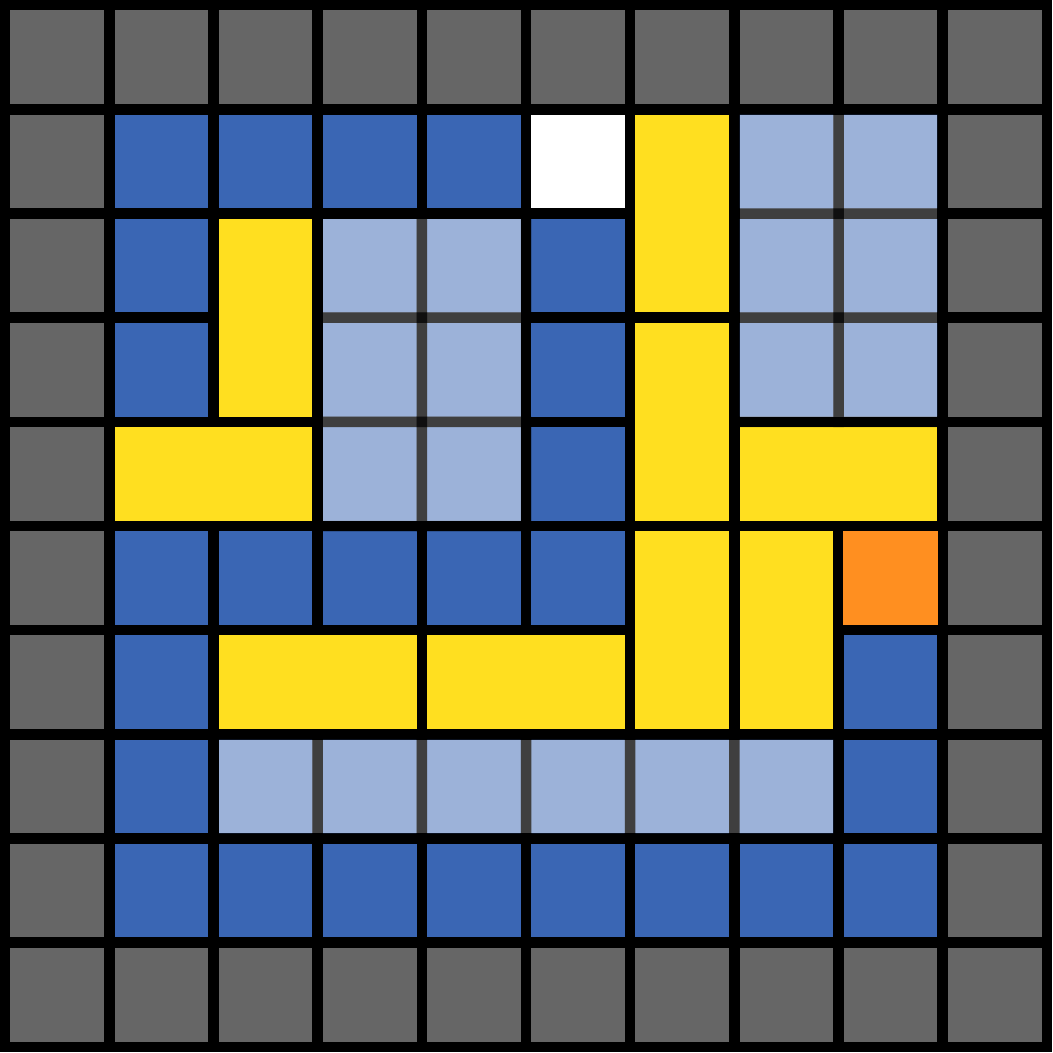}
        \caption{W}
        \label{fig:full2}
    \end{subfigure}
    \begin{subfigure}[b]{0.115\textwidth}
        \centering
        \includegraphics[width = 0.9\textwidth]{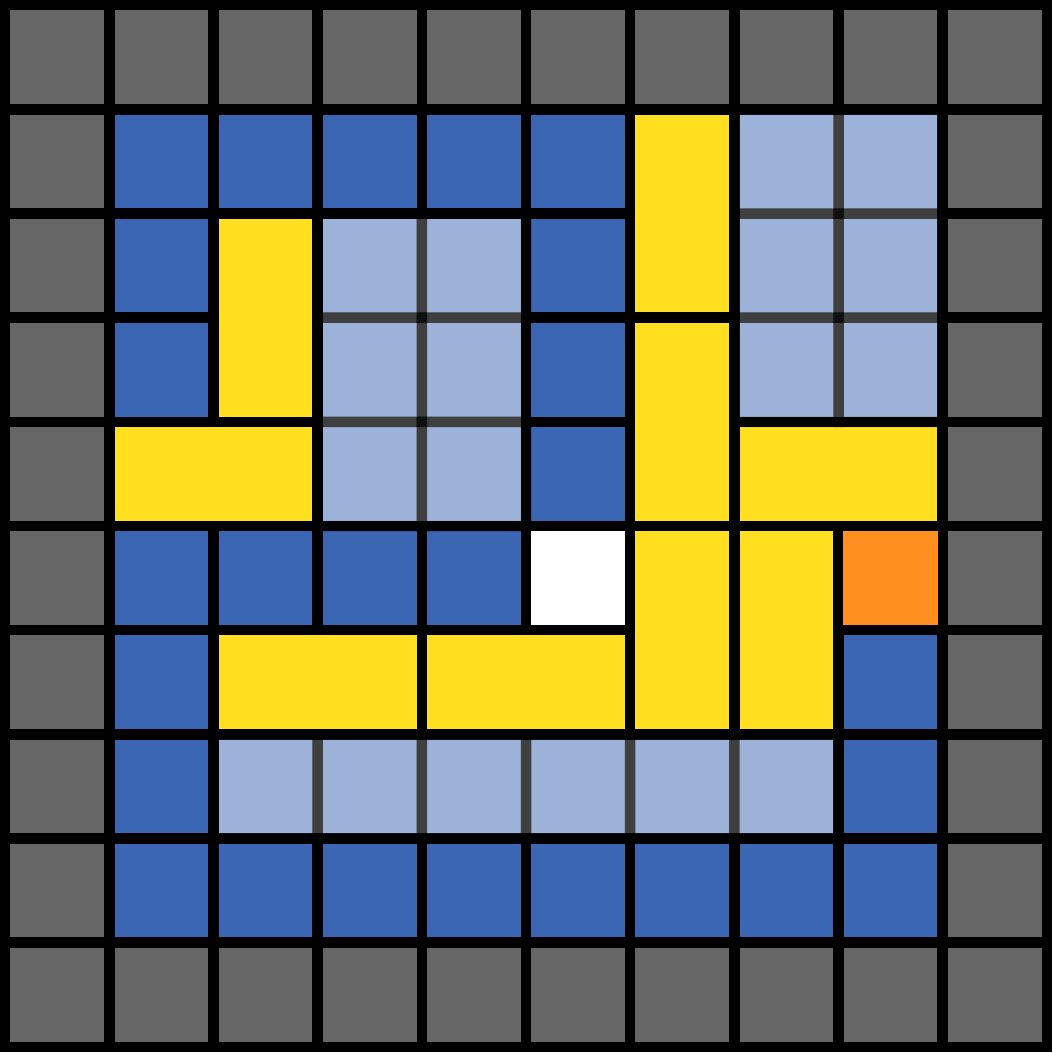}
        \caption{N}
        \label{fig:full3}
    \end{subfigure}
    \begin{subfigure}[b]{0.115\textwidth}
        \centering
        \includegraphics[width = 0.9\textwidth]{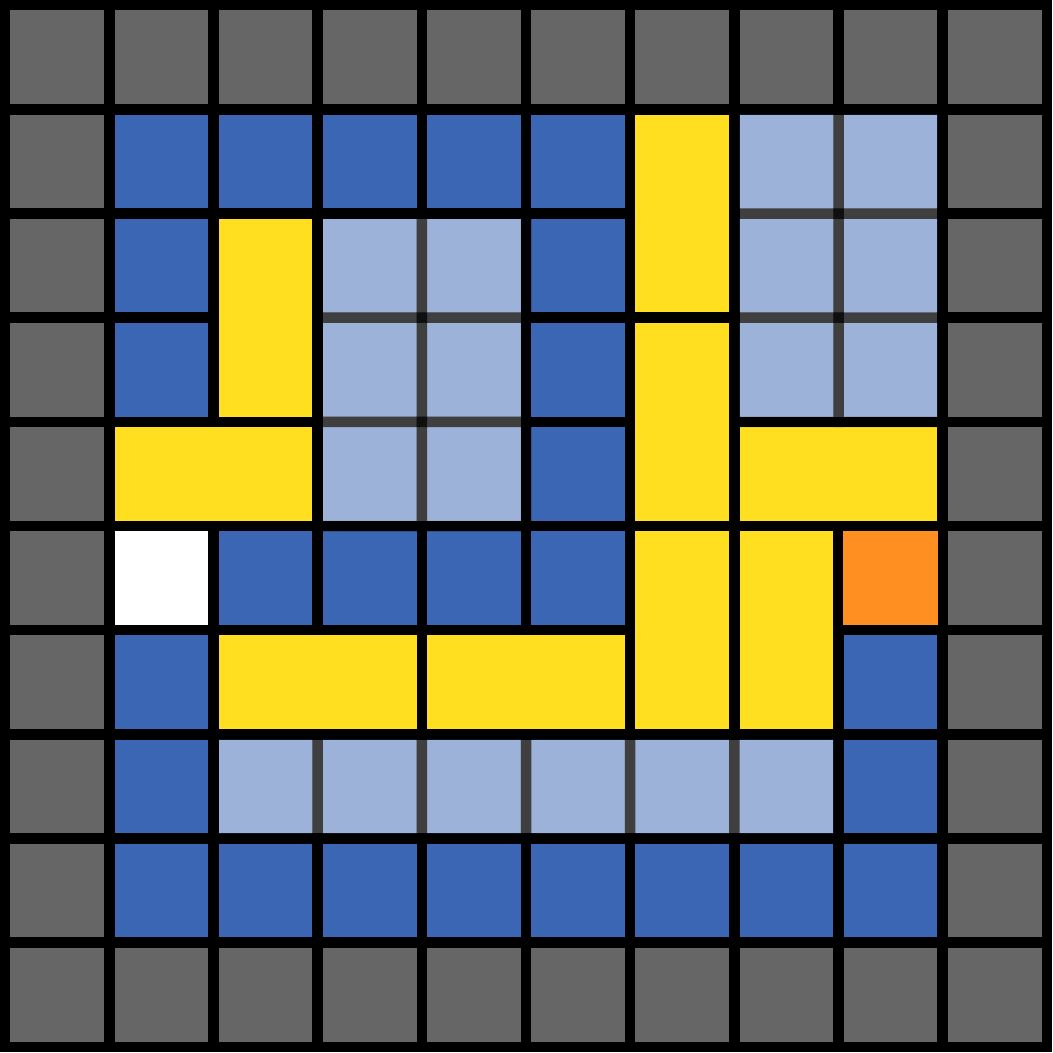}
        \caption{E}
        \label{fig:full4}
    \end{subfigure}
    \begin{subfigure}[b]{0.115\textwidth}
        \centering
        \includegraphics[width = 0.9\textwidth]{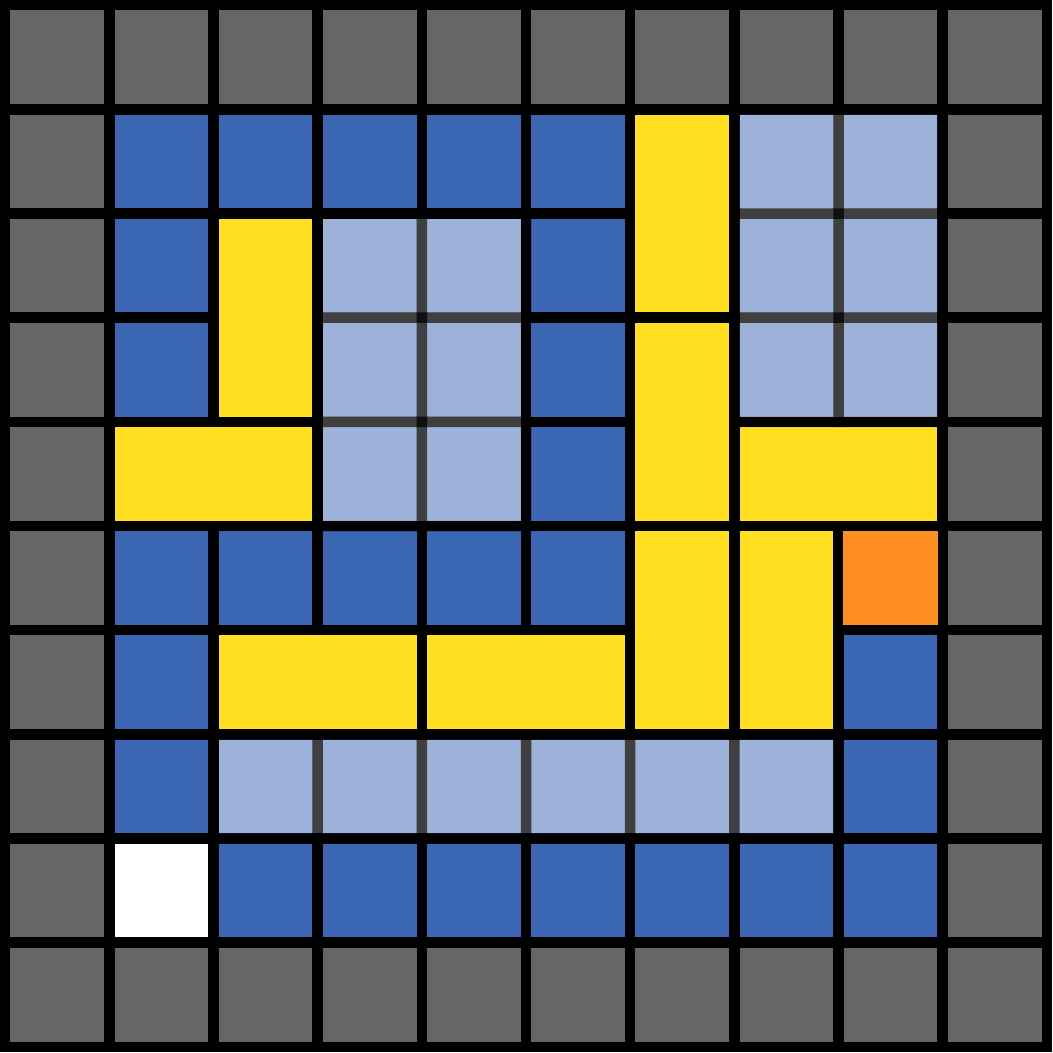}
        \caption{N}
        \label{fig:full5}
    \end{subfigure}
    \begin{subfigure}[b]{0.115\textwidth}
        \centering
        \includegraphics[width = 0.9\textwidth]{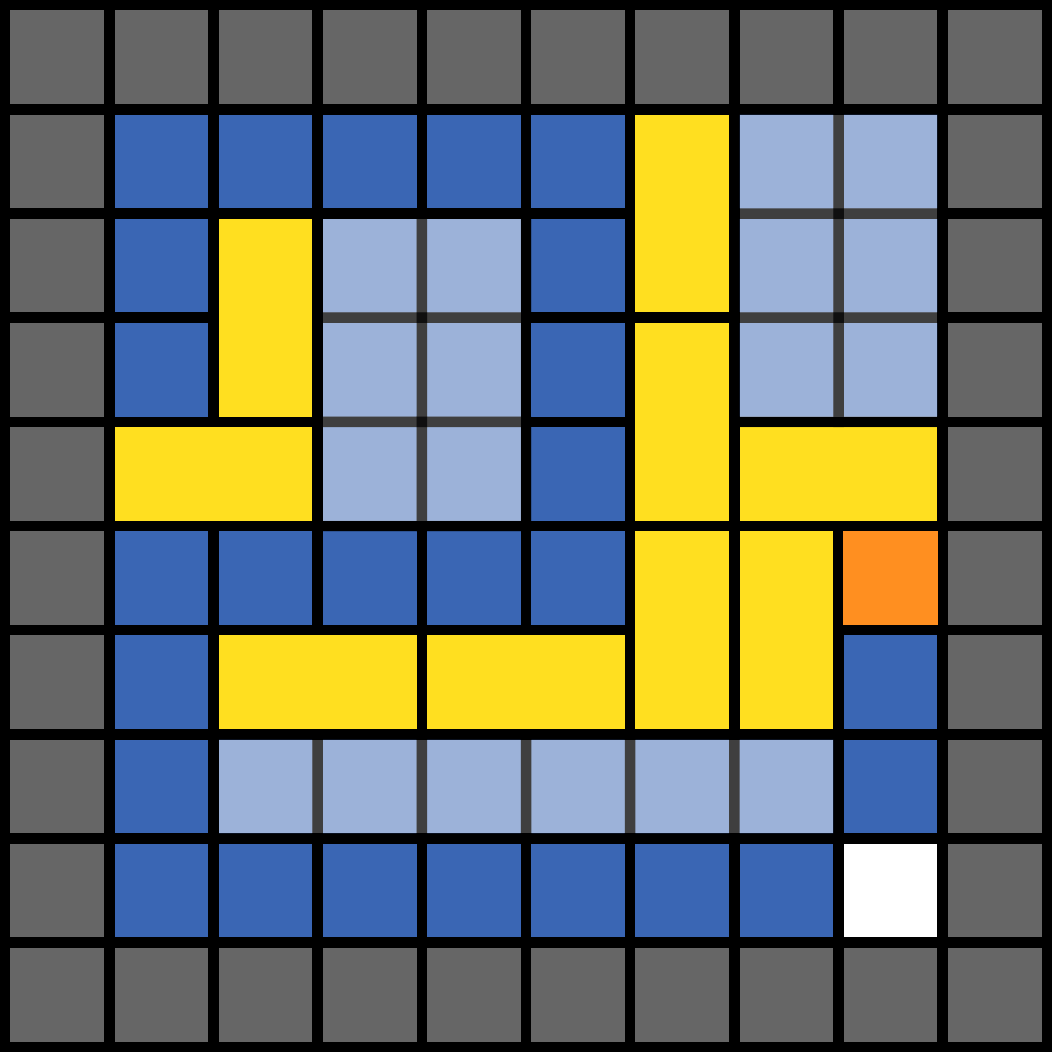}
        \caption{W}
        \label{fig:full6}
    \end{subfigure}
    \begin{subfigure}[b]{0.115\textwidth}
        \centering
        \includegraphics[width = 0.9\textwidth]{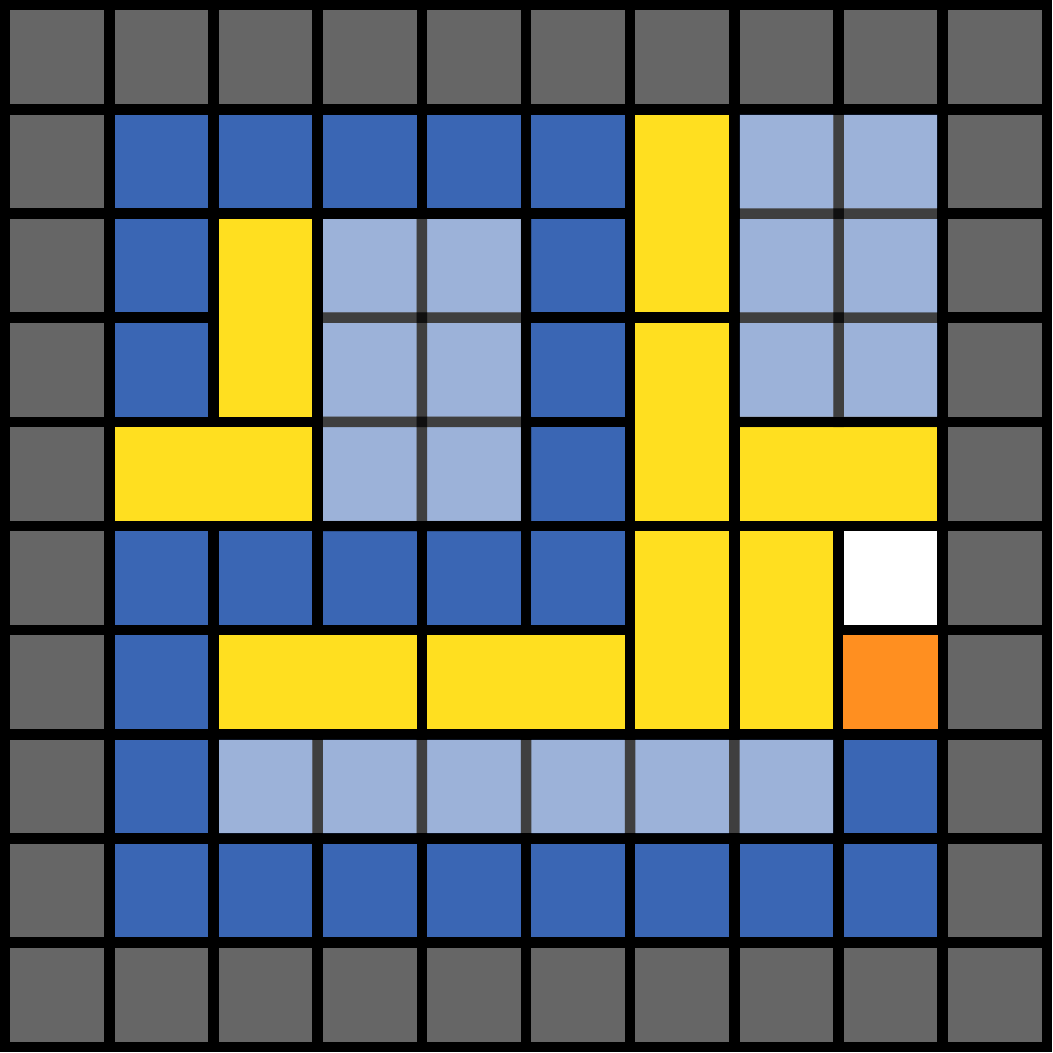}
        \caption{S}
        \label{fig:full7}
    \end{subfigure}
    \caption{An example of a empty space moving through a configuration. The board geometry is just a rectangular frame. The dominoes along with many of the tiles are gridlocked and cannot move. Gridlocked tiles are lighter blue. We can see that through a sequence of tilts the space can move through the configuration and eventually allow the orange tile to change position.}
    \label{fig:fullExample}
\end{figure}

\subsubsection{Construction Preliminaries}
Here we present the construction details for a tilt model gadget which behaves like the CTL gadget described above. We refer to this as the CTL domino gadget.
Unlike the previous constructions, this one uses a rectangular board that is completely filled with polyominoes except for one open position.
In this construction, the open position or ``space" will act as the agent rather than the robot polyomino which was used before.
When a tilt is made, all movable tiles adjacent to the space will move in that direction and fill the space.
The space will travel accordingly until it reaches either the edge of the board or an immovable polyomino.
In this way, we can effectively simulate geometry with gridlocked dominoes.
Along with the CTL domino gadget, we also present several helper gadgets which are needed for the construction. These are the 3-way branching gadget, the straight-wire gadget, the corner turning gadget, the start gadget, and the goal gadget.

\begin{figure}
    \centering
    \begin{subfigure}[b]{0.4\textwidth}
        \centering
        \includegraphics[width = 0.5\textwidth]{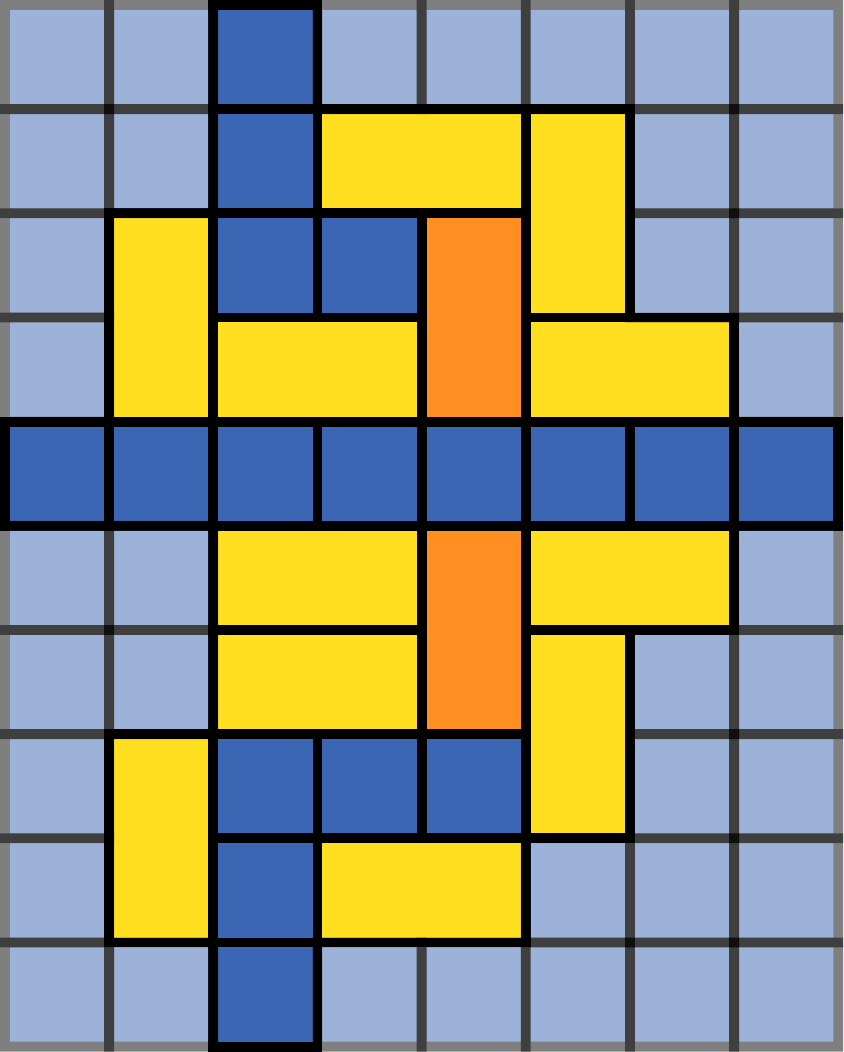}
        \caption{Unlocked CTL Domino Gadget.}
        \label{fig:CTL_unlocked}
    \end{subfigure}
    \begin{subfigure}[b]{0.4\textwidth}
        \centering
        \includegraphics[width = 0.5\textwidth]{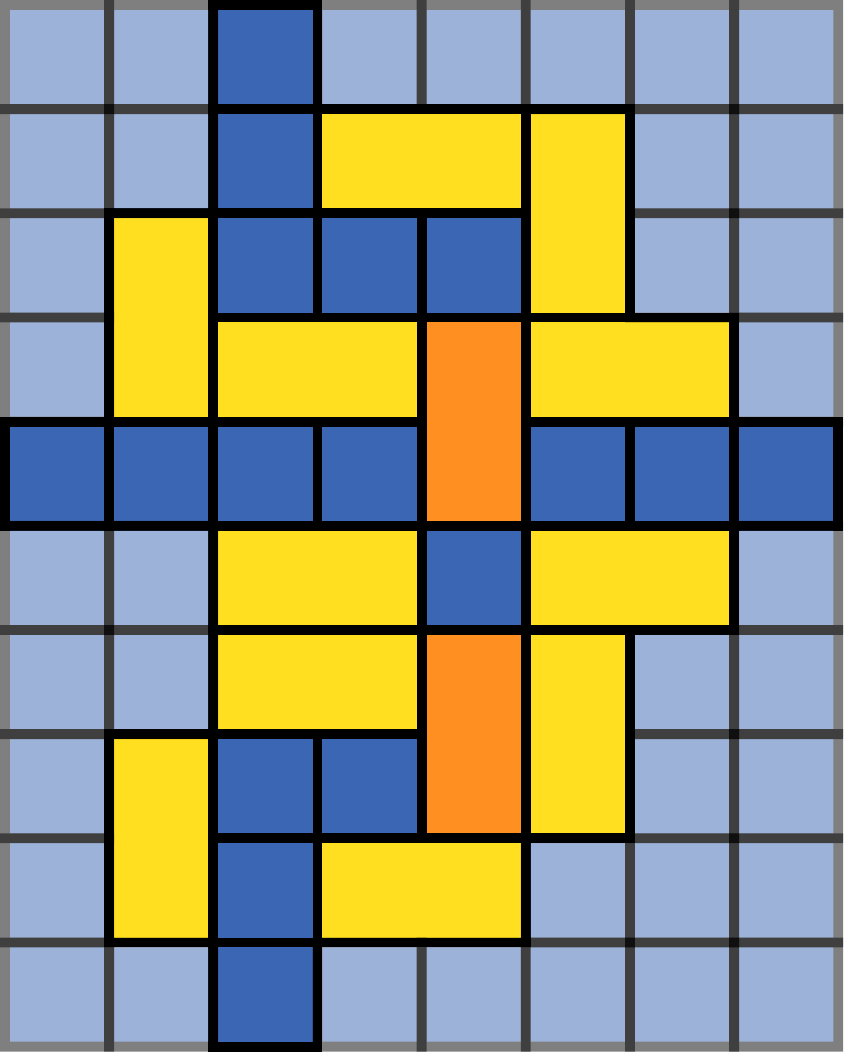}
        \caption{Locked CTL Domino Gadget}
        \label{fig:CTL_locked}
    \end{subfigure}
    \caption{Two states of the CTL domino gadget. Dark blue tiles represent tiles that can move to transfer the space. The light blue tiles and yellow dominoes never move. Orange dominoes can move and the positions of them indicate the state of the gadget with the left being unlocked and right being locked.}
    \label{fig:toggleLock}
\end{figure}

\paragraph{Crossing Toggle-Lock Domino Gadget}
The crossing toggle lock shown in Figure \ref{fig:toggleLock} enforces the traversal rules with two dominoes. In the unlocked state (Figure~\ref{fig:CTL_unlocked}) there is no domino blocking the horizontal tunnel and the space may pass through freely without stopping or changing the state of the gadget. In the locked state (Figure~\ref{fig:CTL_locked}) there is a domino blocking the horizontal tunnel and will prevent the space from traversing in either direction.

When attempting to traverse through the vertical tunnel a domino enforces that it can only be entered through the correct location.
In the unlocked state, the space may only traverse the vertical path from south to north (otherwise it encounters an immovable domino). When the gadget is in the loced state, the space may only travel from north to south.

\paragraph{Wire Gadget}
The wire gadget (Figure~\ref{fig:domino_wire}) is really just a collection of single tiles. These single tiles just act as a the medium through which the agent travels.
Since our agent can only ever be at the beginning or end of a wire, we don't need to introduce any immovable dominoes in the wire gadget.

\paragraph{Corner Gadget}
The agent in the puzzle solvability problem allows for the wires that turn.
We construct a corner gadget (Figure~\ref{fig:domino_turn}) which allows the space to turn corners.
There are two dominos in the gadget which allow the agent to enter from one direction, stop at a domino, and exit the other direction.

\paragraph{Branching Gadget}
Another gadget required for the reduction is a 3-way branching gadget. This gadget allows for the agent to enter from one location and leave the gadget through either of the others. Like the corner gadget, the gridlocked dominoes act as barriers which stop the space from moving.

\paragraph{Start Gadget}
The start gadget (Figure~\ref{fig:domino_start}) is where the agent starts. This gadget contains the open position that acts as the agent. The open position is surrounded by three dominoes to enforce that the agent can only exit in one direction.

\paragraph{Goal Gadget}
Reaching the goal gadget (Figure~\ref{fig:occ}) is the objective of the agent. The gadget contains a horizontal domino which can only be moved if the vertical domino is allowed to move down. So, in order to relocate the horizontal domino, the agent must reach the goal gadget. Also, since none of the yellow dominoes can move because they are permanently gridlocked, the empty location in the goal gadget can only be occupied if the horizontal domino can be relocated.

\begin{figure}
    \centering
    \begin{subfigure}[b]{0.12\textwidth}
        \centering
        \includegraphics[width = \textwidth]{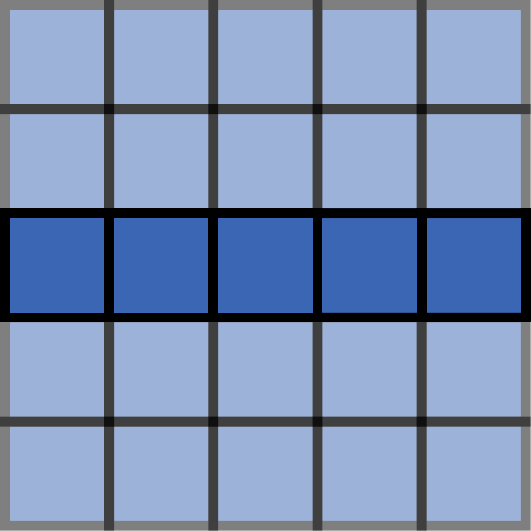}
        \caption{}
        \label{fig:domino_wire}
    \end{subfigure}
    \hspace*{.5cm}
    \begin{subfigure}[b]{0.12\textwidth}
        \centering
        \includegraphics[width = \textwidth]{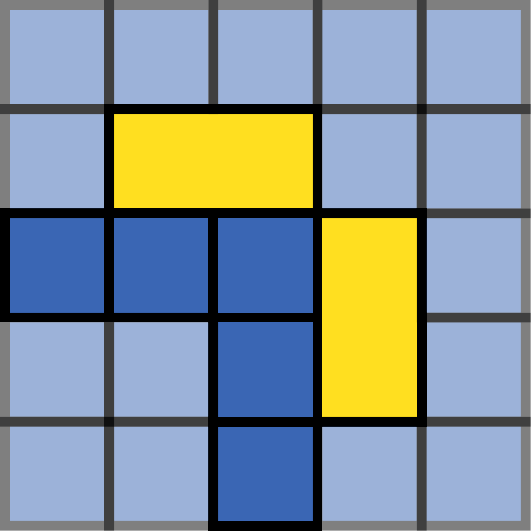}
        \caption{}
        \label{fig:domino_turn}
    \end{subfigure}
    \hspace*{.5cm}
    \begin{subfigure}[b]{0.1535\textwidth}
        \centering
        \includegraphics[width = \textwidth]{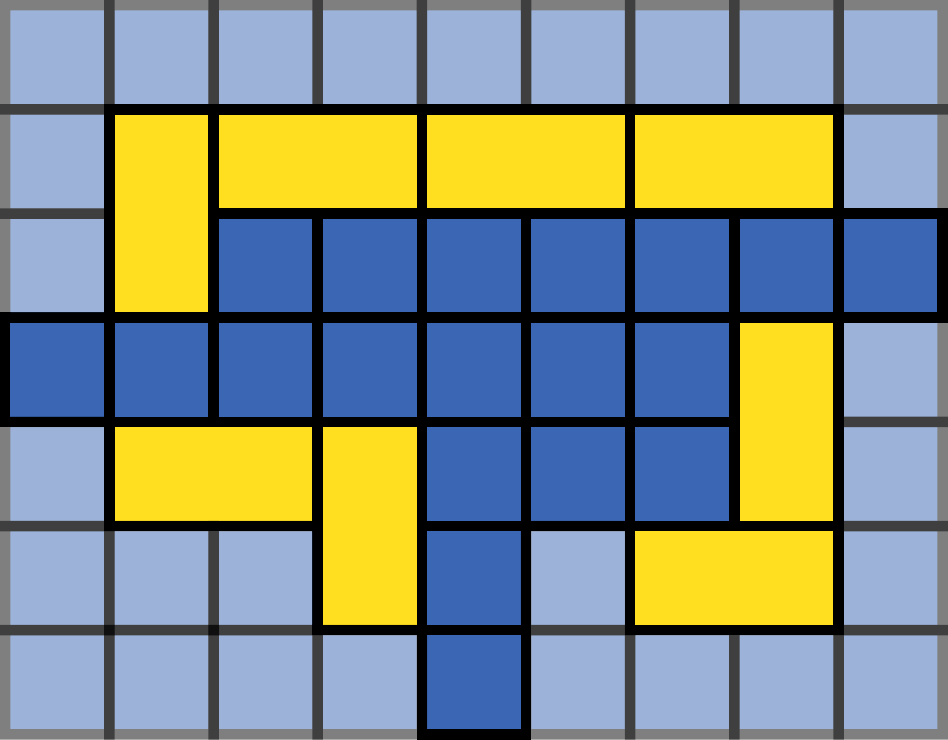}
        \caption{}
        \label{fig:branch}
    \end{subfigure}
    \hspace*{.5cm}
    \begin{subfigure}[b]{0.12\textwidth}
        \centering
        \includegraphics[width = \textwidth]{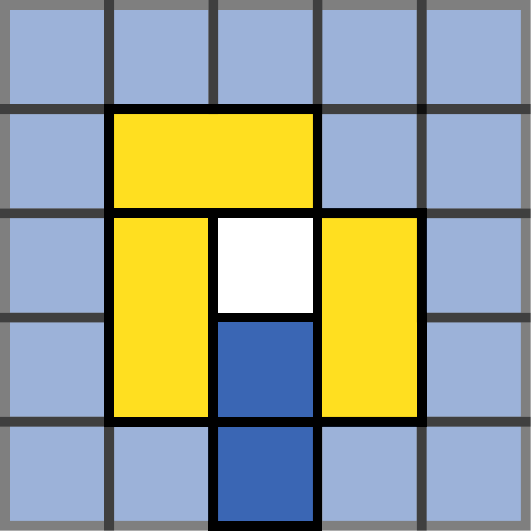}
        \caption{}
        \label{fig:domino_start}
    \end{subfigure}
    \hspace*{.5cm}
    \begin{subfigure}[b]{0.12\textwidth}
        \centering
        \includegraphics[width = \textwidth]{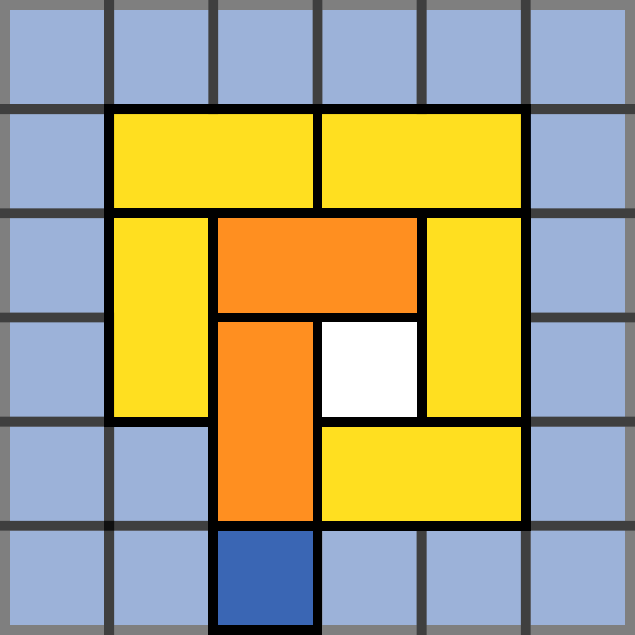}
        \caption{}
        \label{fig:occ}
    \end{subfigure}
    \caption{(a) The Straight Wire Gadget used to allow the agent to traverse straight wires. (b) The Corner Gadget used to allow the agent to change direction. (c) 3-way Branching Gadget. Allows for a 3-way direction change.
    (d) The Start Gadget which contains the open position that acts as the agent. (e) The Goal Gadget. Moving the horizontal orange domino is the objective of the agent.}
    \label{fig:domino_gadgets}
\end{figure}

\subsubsection{Limited Geometry PSPACE-complete Results}

\begin{lemma}\label{lem:CTL_domino}
  The CTL domino gadget correctly implements the behavior of the CTL gadget.
\end{lemma}

\begin{proof}
  Like the CTL gadget, the CTL domino gadget has two tunnels. Observe that the CTL domino gadget has two states, as does the CTL gadget.
  Observe that in the unlocked state, an agent (an empty space) may pass through the horizontal tunnel in either direction, while agent traversal through the vertical tunnel can only be done from south to north.
  Also, observe that in the locked state, an agent cannot pass through the horizontal tunnel (as it is blocked by a domino), while traversal through the vertical tunnel can only be done from north to south.
  Lastly, observe that any traversal of the vertical tunnel toggles the gadget between its locked and unlocked state.
  Thus, the CTL domino gadget correctly implements the behavior of the CTL gadget.
\end{proof}

\begin{theorem}\label{thm:PSPACEdomi}
The relocation problem is PSPACE-complete when limited to a rectangular board and allowing for $1\times1$ tiles, and $1\times2$, $2\times1$ dominoes.
\end{theorem}

\begin{proof}
  We have already shown PSPACE membership for the relocation problem in Theorem~\ref{thm:relocation}. We now show the problem is PSPACE-hard even with a rectangular board filled with $1\times1$ tiles, and $1\times2$, $2\times1$ dominoes by reducing from the CTL puzzle solvability problem.

  Given an instance of the CTL puzzle solvability problem from \cite{DBLP:conf/fun/DemaineGLR18}, used the tools described above to create a tilt model configuration as follows:
  For each CTL gadget in the puzzle, create a crossing-toggle lock domino gadget such that the initial state of the domino gadget matches the initial state of the CTL gadget.
  For each straight wire in the puzzle, create a wire gadget in the tilt configuration.
  For each 3-way intersection in the puzzle, create a branching gadget in the tilt configuration.
  For each wire turn in the puzzle, create a corner gadget in the tilt configuration.
  For the start and goal locations in the puzzle, create a start gadget and goal gadget in the tilt configuration.
  Create the board geometry as a rectangle of blocked spaces that surrounds the configuration.
  Fill all remaining empty board locations (except the open space in the start and goal gadgets) with $1\times 1$ tiles.

  Lemma~\ref{lem:CTL_domino}, along with the fact that there clearly exists move sequences to traverse the other gadgets, shows that if the given CTL puzzle has a solution, then there exists a tilt sequence which moves the space to the goal gadget in our tilt configuration.

  Furthermore, since there only exists a single space in the configuration (exluding the stationary space in the goal gadget), Lemma~\ref{lem:CTL_domino} also shows that no solution for the CTL puzzle implies no solution for our relocation problem.
\end{proof}

\begin{corollary}\label{thm:occupancy}
The occupancy problem is PSPACE-complete when limited to a rectangular board and allowing for $1\times1$ tiles, and $1\times2$, $2\times1$ dominoes.
\end{corollary}
\begin{proof}
This follows from the same construction that was used in Theorem \ref{thm:PSPACEdomi}. The agent can only reach the goal if the puzzle if a solution exists to the puzzle solvability problem. The space in the goal gadget can only be occupied if the agent reaches the gadget.
\end{proof}

\section{Future Work} \label{Conclusion}
There are a number of open questions and directions for future work that stem from our results.  In the area of complexity, we have shown that relocation, occupancy, and reconfiguration problems are PSPACE-complete if both $1\times 1$ tiles and a single $2\times 2$ block are considered, and the relocation and occupancy problems are PSPACE-complete when limited to a rectangular board and allowing $1\times 1$ and $1\times 2$ and $2\times 1$ dominos.  The complexity of this problem is only known to be NP-hard if restricted to $1\times 1$ movable blocks~\cite{BecDemFek2017PCCAL}.  In~\cite{BecDemFek2017PCCAL} the authors showed that $1\times 1$ blocks have substantial limitations such as the impossibility of implementing a certain type of fanout gate.  This might be evidence that this variant is not PSPACE-complete.  However, the existence of necessarily exponentially long tilt sequences for reconfiguration shown in~\cite{BecDemFek2017PCCAL} provides some evidence that the problem does not lie within the class NP.  An alternate set of questions involve the trade off between of number of larger than $1 \times 1$ tiles and the board complexity.  Does the problem remain PSPACE-complete when both are restricted, or does the problem become simpler, perhaps allowing for a polynomial time solution?

Another direction of future work is to explore strong universal configurations for classes beyond the ``drop'' class.  One plausible extension might entail building drop shapes as a subroutine and combining them together in a hierarchical ``staged assembly'' fashion. Another could be to employ a ``sand sifter'' along with this staged assembly as a part of the pattern builder. Related to this is the consideration of speeding up the assembly process by attempting to build distinct portions of a target shape in parallel.  This approach has been recently explored~\cite{schmidt2018efficient}, but has not been considered in the context of building universal configurations.

One interesting direction for future work involves relaxing the constraint in which polyominos slide \emph{maximally} until stopped.  Instead, polyominos could slide a fixed amount per tilt, or travel at some particular speed.  This strengthens the model significantly, but is motivated by a number of practical proposed implementations.  What interesting added capabilities does this modification allow, and how do complexity questions change for this model?

\section*{Acknowledgments} We would like to thank Andrew Winslow for pointing us to the literature on the C2T gadget, which greatly simplified our proof of the complexity of the reconfiguration and relocation problems.  We would also like to thank Aaron Becker for providing feedback on the final presentation of this work.  Finally, we would like to thank the anonymous reviewers of the preliminary conference version of this paper for their thorough, careful, and constructive feedback.
\bibliographystyle{amsplain}
\bibliography{tam,tilt}

\providecommand{\bysame}{\leavevmode\hbox to3em{\hrulefill}\thinspace}
\providecommand{\MR}{\relax\ifhmode\unskip\space\fi MR }
\providecommand{\MRhref}[2]{%
  \href{http://www.ams.org/mathscinet-getitem?mr=#1}{#2}
}
\providecommand{\href}[2]{#2}
\begin{thebibliography}{10}

\bibitem{atomixGame}
\emph{Atomix}, Thalion Software, 1990.

\bibitem{megaMazeGame}
\emph{Mega maze}, Phillips Media, 1995.

\bibitem{FullTilt2019}
Jose Balanza-Martinez, Austin Luchsinger, David Caballero, Rene Reyes, Angel
  Cantu, Robert Schweller, Luis Garcia, and Tim Wylie, \emph{Full tilt:
  Universal constructors for general shapes with uniform external forces},
  Thirtieth Annual ACM-SIAM Symposium on Discrete Algorithms, SODA '19, 2019.

\bibitem{tiltVideoBecker}
Aaron~T. Becker, \emph{Parallel self-assembly of polyominoes under uniform
  control inputs}, \url{https://www.youtube.com/watch?v=G93H1Tecj-w}, 2018.

\bibitem{BecDemFek2014RMPS}
Aaron~T. Becker, Erik~D. Demaine, S{\'a}ndor~P. Fekete, Golnaz Habibi, and
  James McLurkin, \emph{Reconfiguring massive particle swarms with limited,
  global control}, Algorithms for Sensor Systems (Berlin, Heidelberg), Springer
  Berlin Heidelberg, 2014, pp.~51--66.

\bibitem{BecDemFek2017PCCAL}
Aaron~T. Becker, Erik~D. Demaine, S{\'a}ndor~P. Fekete, Jarrett Lonsford, and
  Rose Morris-Wright, \emph{Particle computation: complexity, algorithms, and
  logic}, Natural Computing (2017).

\bibitem{BecDemFek2014PCRA}
Aaron~T. Becker, Erik~D. Demaine, S{\'{a}}ndor~P. Fekete, and James McLurkin,
  \emph{Particle computation: Designing worlds to control robot swarms with
  only global signals},  (2014), 6751--6756.

\bibitem{BecFek2017TAMF}
Aaron~T. Becker, S{\'a}ndor~P. Fekete, Phillip Keldenich, Dominik Krupke,
  Christian Rieck, Christian Scheffer, and Arne Schmidt, \emph{{Tilt Assembly:
  Algorithms for Micro-Factories that Build Objects with Uniform External
  Forces}}, 28th International Symposium on Algorithms and Computation (ISAAC
  2017), Leibniz International Proceedings in Informatics (LIPIcs), vol.~92,
  2017, pp.~11:1--11:13.

\bibitem{Becker2013FCMM}
Aaron~T. Becker, Yan Ou, Paul Kim, Min~J. Kim, and Agung Julius, \emph{Feedback
  control of many magnetized: Tetrahymena pyriformis cells by exploiting phase
  inhomogeneity}, 2013 IEEE/RSJ International Conference on Intelligent Robots
  and Systems, Nov 2013, pp.~3317--3323.

\bibitem{labyrinthGame}
Brio, \emph{Labyrinth game.}

\bibitem{DBLP:conf/fun/DemaineGLR18}
Erik~D. Demaine, Isaac Grosof, Jayson Lynch, and Mikhail Rudoy,
  \emph{Computational complexity of motion planning of a robot through simple
  gadgets}, 9th International Conference on Fun with Algorithms, {FUN} 2018,
  June 13-15, 2018, La Maddalena, Italy, 2018, pp.~18:1--18:21.

\bibitem{Doty-2012a}
Dave Doty, \emph{Theory of algorithmic self-assembly}, Communications of the
  ACM \textbf{55} (2012), no.~12, 78--88.

\bibitem{doi:10.1021/acsnano.7b02693}
Cheulhee Jung, Peter~B. Allen, and Andrew~D. Ellington, \emph{A simple, cleated
  dna walker that hangs on to surfaces}, ACS Nano \textbf{11} (2017), no.~8,
  8047--8054, PMID: 28719175.

\bibitem{2DDiscPoly}
Phillip Keldenich, Sheryl Manzoor, Li~Huang, Dominik Krupke, Arne Schmidt,
  Sandor~P. Fekete, and Aaron~T. Becker, \emph{On designing 2{D} discrete
  workspaces to sort or classify 2{D} polyominoes}, 2018 IEEE/RSJ International
  Conference on Intelligent Robots and Systems (IROS).

\bibitem{khali2016controlmp}
Islam~S.M. Khalil, Hazem Abass, Mostafa Shoukry, Anke Klingner, Rasha~M.
  El-Nashar, Mohamed Serry, and Sarthak Misra, \emph{Robust and optimal control
  of magnetic microparticles inside fluidic channels with time-varying flow
  rates}, International Journal of Advanced Robotic Systems \textbf{13} (2016),
  no.~3, 123.

\bibitem{martel2009mri}
Sylvain Martel, Ouajdi Felfoul, Jean-Baptiste Mathieu, Arnaud Chanu, Samer
  Tamaz, Mahmood Mohammadi, Martin Mankiewicz, and Nasr Tabatabaei,
  \emph{{MRI}-based medical nanorobotic platform for the control of magnetic
  nanoparticles and flagellated bacteria for target interventions in human
  capillaries}, The International journal of robotics research \textbf{28}
  (2009), no.~9, 1169--1182.

\bibitem{minecraft}
Mojang, \emph{Minecraft}.

\bibitem{nacoTileSurvey}
Matthew~J. Patitz, \emph{An introduction to tile-based self-assembly and a
  survey of recent results}, Natural Computing \textbf{13} (2014), no.~2,
  195--224.

\bibitem{schmidt2018efficient}
Arne Schmidt, Sheryl Manzoor, Li~Huang, Aaron~T. Becker, and Sándor Fekete,
  \emph{Efficient parallel self-assembly under uniform control inputs}, IEEE
  Robotics and Automation Letters (2018), 1--1.

\bibitem{Sinha1909}
Rajni Sinha, Gloria~J. Kim, Shuming Nie, and Dong~M. Shin, \emph{Nanotechnology
  in cancer therapeutics: bioconjugated nanoparticles for drug delivery},
  Molecular Cancer Therapeutics \textbf{5} (2006), no.~8, 1909--1917.

\bibitem{ThinkfunTilt}
ThinkFun, \emph{Tilt: Gravity fed logic maze.}

\bibitem{Woods20140214}
Damien Woods, \emph{Intrinsic universality and the computational power of
  self-assembly}, Philosophical Transactions of the Royal Society of London A:
  Mathematical, Physical and Engineering Sciences \textbf{373} (2015),
  no.~2046.

\bibitem{strand-displacement:2011}
David Zhang and Georg Seelig, \emph{Dynamic dna nanotechnology using
  strand-displacement reactions},  \textbf{3} (2011), 103--13.

\bibitem{Zhou2015APN}
Chao Zhou, Xiaoyang Duan, and Na~Liu, \emph{A plasmonic nanorod that walks on
  dna origami}, Nature Communications \textbf{6} (2015), 8102--8102,
  26303016[pmid].

\end{thebibliography}

\end{document}